\documentclass{article}

\usepackage{caption}

\newcommand{\ours}{GeoCoupling}

\newcounter{mycounter}

\newcommand{\ul}[1]{{\underline{#1}}}

\usepackage{paralist}
\usepackage{tablefootnote}
\usepackage{multirow, multicol}

\usepackage{amsmath,amsfonts,bm}

\def\eqref#1{equation~\ref{#1}}

\def\1{\bm{1}}

\def\ry{{\textnormal{y}}}

\def\rve{{\mathbf{e}}}

\def\rvu{{\mathbf{i}}}

\def\rvt{{\mathbf{t}}}
\def\rvu{{\mathbf{u}}}

\def\rvx{{\mathbf{x}}}

\def\vh{{\bm{h}}}

\def\vr{{\bm{r}}}

\DeclareMathAlphabet{\mathsfit}{\encodingdefault}{\sfdefault}{m}{sl}
\SetMathAlphabet{\mathsfit}{bold}{\encodingdefault}{\sfdefault}{bx}{n}

\newcommand{\E}{\mathbb{E}}

\newcommand{\Var}{\mathrm{Var}}

\usepackage[ruled, vlined, boxed]{algorithm2e}
\usepackage[table]{xcolor}

\usepackage{microtype}
\usepackage{graphicx}
\usepackage{subcaption}
\usepackage{booktabs} 

\usepackage{hyperref}

\usepackage[accepted]{icml2026}

\usepackage{amsmath}
\usepackage{amssymb}
\usepackage{mathtools}
\usepackage{amsthm}

\usepackage[capitalize,noabbrev]{cleveref}

\theoremstyle{plain}
\newtheorem{theorem}{Theorem}[section]
\newtheorem{proposition}[theorem]{Proposition}

\theoremstyle{definition}
\newtheorem{definition}[theorem]{Definition}

\theoremstyle{remark}

\usepackage[textsize=tiny]{todonotes}

\icmltitlerunning{Demystifying Multimodal Biomolecular Co-design with Intrinsic Geodesic Coupling}

\begin{document}

\twocolumn[
  \icmltitle{Demystifying Multimodal Biomolecular Co-design \\ with Intrinsic Geodesic Coupling}



  \icmlsetsymbol{equal}{*}

  \begin{icmlauthorlist}
\icmlauthor{Keyue Qiu}{equal,air,thu}
\icmlauthor{Xintong Wang}{equal,thu}
\icmlauthor{Zhilong Zhang}{air,qzc}
\icmlauthor{Hao Zhou}{air,shanghai}
\icmlauthor{Wei-Ying Ma}{air}
  \end{icmlauthorlist}

\icmlaffiliation{thu}{Department of Computer Science and Technology, Tsinghua University}
\icmlaffiliation{air}{Institute for AI Industry Research (AIR), Tsinghua University}
\icmlaffiliation{shanghai}{Shanghai Artificial Intelligence Laboratory}
\icmlaffiliation{qzc}{Qiuzhen College, Tsinghua University}

\icmlcorrespondingauthor{Hao Zhou}{zhouhao@air.tsinghua.edu.cn}

  \icmlkeywords{Machine Learning, ICML}

  \vskip 0.3in
]



\printAffiliationsAndNotice{\icmlEqualContribution}

\begin{abstract}
Biomolecules such as proteins and small-molecule ligands play a central role in biological systems, arising from the tight interplay between sequence and three-dimensional structure.
Recent generative models for biomolecular co-design aim to capture this interplay by jointly modeling coupled modalities.
However, existing approaches largely adopt a parallel execution of marginal generative processes, implicitly enforcing fixed synchronous coupling.
We argue that a critical but overlooked degree of freedom lies in how these marginal processes are \emph{temporally coupled} during training and generation, where inappropriate coupling can introduce high-variance supervision and inconsistent intermediate states, affecting modality consistency.
To address this, we introduce \ours, a systematic framework that optimizes for temporal couplings between heterogeneous modalities.
Empirical results across structure-based drug design and unconditional protein design demonstrate the learned couplings consistently outperform synchronous and randomly coupled baselines, yielding biomolecules with improved physical validity and diversity.

\end{abstract}

\section{Introduction}


Functional biomolecules, ranging from proteins to small-molecule ligands, are the fundamental building blocks of life, governing biological processes from metabolic catalysis to signal transduction \cite{branden2012introduction}. Their biological utility is intrinsically linked to the relationship between sequences and three-dimensional structures. Consequently, the joint design of biomolecular sequences and structures has emerged as a central challenge in computational biology, with transformative potential for protein engineering and de novo drug design \cite{huang2016coming, watson2023novo}.


Recently, deep generative models have made remarkable progress in biomolecular design. Early approaches focused on unimodal generation for sequences \cite{madani2020progen, nijkamp2023progen2, wang2024diffusion} or structures \cite{watson2023novo, yim2023se3, bosefoldflow}. However, practical discovery tasks such as structure-based drug design (SBDD) \cite{schneuing2024structure} require determining sequence and structure jointly and capturing holistic cross-modal dependencies. 
Adapting these unimodal approaches for joint generation relies on sequential pipelines that chained separate tools for structure prediction \citep{jumper2021highly, abramson2024accurate} or inverse folding \cite{dauparas2022robust}.
Although effective as heuristics, these disjoint workflows fail to exploit the biological insight such as induced fit in ligand binding or the co-adaptation between backbone geometry and sequence identity. 
To address this, recent works have shifted toward multimodal co-design, treating sequence and structure as coupled modalities within a unified generative process.



A central challenge in multimodal co-design lies in the \emph{dynamic mismatch} between heterogeneous modalities. Different biological modalities exhibit distinct intrinsic characteristics: continuous 3D coordinates evolve on high-dimensional geometric manifolds, whereas sequences reside in combinatorial discrete space. Most existing co-design approaches for coupled modalities impose a \emph{synchronous coupling}, forcing all modalities to evolve from noise to data at shared timestep. While convenient, this assumption introduces a restrictive inductive bias by ignoring the dependency structure and asymmetric convergence behaviors across modalities. Empirically, one modality may stabilize earlier than another, causing the model to condition on poorly resolved or inconsistent representations during intermediate stages. This misalignment forces the model to learn from off-manifold intermediate states, leading to noisy or even conflicting gradient signals, and degraded modality consistency for biomolecules.



\begin{figure*}
    \centering
    \includegraphics[width=\linewidth]{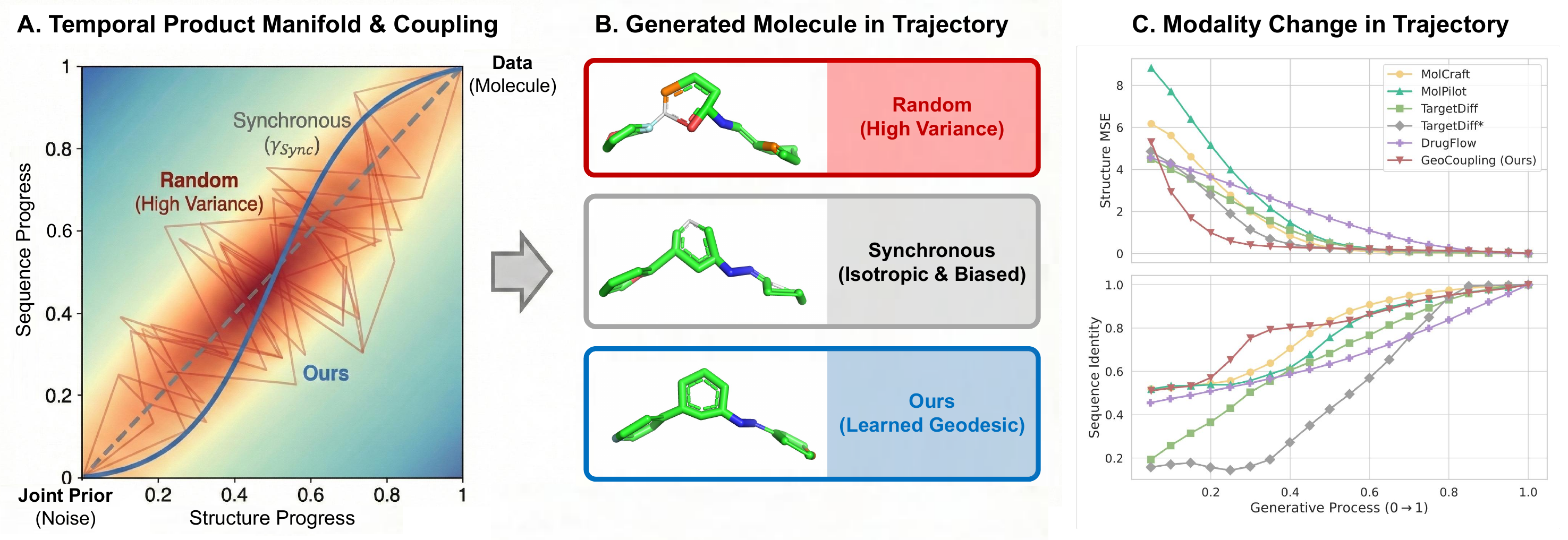}
    \caption{\textbf{Decoupling Multimodal Generation via Learned Geodesics.} \textbf{A.} Temporal Product Manifold: Visualization of coupling strategies on the joint noise space. While Synchronous coupling (gray) forces an isotropic diagonal path and Random coupling (red) suffers from high variance , GeoCoupling (blue) discovers a learned geodesic
  that navigates the optimal transport path. \textbf{B.} Conceptual illustration for the intermediate states of generated molecules, where random couplings often lead to crossmodal inconsistency, whereas our learned path ensures geometric validity. \textbf{C.} Trajectory Analysis: Quantitative evolution of Structure MSE (top) and Sequence Identity (bottom) during generation ($0\to 1$). GeoCoupling (red triangles) accelerates convergence compared to baselines. Notably, the rapid drop in structure error ($t < 0.3$) precedes the surge in sequence identity, confirming a structure-first generative mechanism.}
    \label{fig:main}
\end{figure*}

Recent work has attempted to relax strict synchronization through random coupling strategies for protein data \citep{pmlr-v235-campbell24a, wang2025dplm2}. While such approaches broaden the training distribution, they often exacerbate training variance and introduce a training--inference mismatch: models are supervised on decoupled noise combinations that are rarely encountered during synchronized generation. Notably, \citet{qiupiloting} showed that optimizing coupling schedules at test time can improve model performance in structure-based drug design. However, such post-hoc procedures require repeated evaluation across modality combinations, leading to rapidly growing computational overhead, while leaving the instability of the training dynamics fundamentally unaddressed. These observations suggest that the limitation of existing co-design methods fundamentally relates to its multimodal training dynamics, where biomolecular co-design operates on a heterogeneous product manifold. On this manifold, the optimal generative trajectory is rarely the diagonal path that synchronizes all modalities uniformly. Instead, effective generation usually follows a dynamically curved path, i.e., a \emph{geodesic} defined by the transport cost of the generative flow—along which different components mature at rates dictated by their respective learning complexities.

In this work, we present \textsc{GeoCoupling}, a framework that treats temporal coupling as a learnable degree of freedom. We reformulate multimodal generation as a \emph{Temporal Optimal Transport} problem, seeking the schedule that minimizes the energy required to transport the prior to the data distribution. Finding this optimal schedule presents a unique optimization challenge: the objective function is the final validation performance after training, making it non-differentiable with respect to the schedule and expensive to query. Crucially, backpropagating through the entire training trajectory to learn the schedule via hypergradients is computationally intractable. To address this, \textsc{GeoCoupling} employs a bi-level optimization strategy with a Gaussian Process surrogate to efficiently navigate the coupling space. This design enables the automatic discovery of complexity-aware curricula that align generative progress with modality-specific difficulty, while incurring negligible additional cost at training and inference time.


Our contributions are as follows:
\begin{itemize}
    \item We identify dynamic mismatch between heterogeneous modalities as a fundamental limitation of existing co-design models, showing that synchronous and random couplings are prone to high bias and high variance, respectively.
    \item We propose a model-agnostic Bayesian optimization framework that learns low-energy temporal couplings during training, providing a principled alternative to heuristic schedules for multimodal generative models.
    \item We demonstrate that the learned couplings consistently outperform synchronous and asynchronous baselines on challenging benchmarks, including structure-based drug design and unconditional protein generation, producing biomolecules with improved physical validity and structural diversity.
\end{itemize}

\section{Related Works}

\paragraph{Protein Design}
Existing approaches can be broadly categorized into three paradigms.
(1) \textit{Unimodal generative models} focus on a single modality: sequence-based methods include autoregressive models such as ProGen \cite{madani_progen_2020, nijkamp2023progen2} and diffusion language model DPLM \cite{wang2024diffusion}, while structure-based models are dominated by non-autoregressive generation such as RFdiffusion \cite{watson2023novo}, FrameDiff \cite{yim2023se3} and FoldFlow \citep{bosefoldflow}.
(2) \textit{Cross-modal conditional generation} tackles translation between modalities: sequence-to-structure tasks (folding) are addressed by AlphaFold \cite{jumper2021highly}, whereas structure-to-sequence tasks (inverse folding) are handled by methods such as ProteinMPNN \cite{dauparas2022robust} and ESM-IF \cite{hsu_learning_2022}.
(3) These successes have unlocked the potential for \textit{multi-modal joint generation}, treating sequence and structure as coupled modalities to enable holistic co-design. Representative methods include models MultiFlow \cite{pmlr-v235-campbell24a} that operate directly in the sample space, La-Proteina \cite{geffner2026laproteina} that relies on VAE for the protein latent space, and DPLM-2 \cite{wang2025dplm2} that utilizes a structure tokenizer. Within the co-design paradigm, the fundamental challenge lies in capturing the intricate relationship between sequence and structure and ensuring better mutual consistency.

\paragraph{Structure-based Drug Design}
Structure-Based Drug Design (SBDD) focuses on generating high-affinity ligands within specific protein binding pockets. Early autoregressive approaches, such as Pocket2Mol \cite{peng_pocket2mol_2022}, employed a step-by-step atom growth strategy. Subsequently, geometric diffusion models \cite{schneuing2024structure, guan_3d_2023} achieved joint denoising of both atomic types and coordinates. MolCRAFT \cite{qu2024molcraft} introduced Bayesian Flow Networks (BFN) \cite{graves2023bayesian} to generate atomic types and positions within a unified diffusion probabilistic framework in parameter space. Building on this, MolPilot \cite{qiupiloting} incorporated Variational Lower Bound Optimal Scheduling (VOS) to navigate multimodal path dependencies. Given that atomic types and positions reside in discrete and continuous modalities respectively, formulating an optimal joint generative path remains a pivotal challenge.

\paragraph{Optimal Transport.}
Optimal Transport (OT) provides a rigorous geometric framework for comparing probability distributions and transporting mass with minimal cost~\citep{villani2008optimal}.
Recent flow-based generative models~\citep{lipman2022flow, liu2022flow} leverage this dynamic view to rectify flow paths, resulting in straight sample trajectories. \citet{song2024equivariant} proposes equivariant optimal transport (EOT), advancing the concept of sample distance in the field of molecular generation.
Existing works largely focus on \textbf{sample coupling} (i.e., pairing $x_0$ and $x_1$	
  in the data space) to reduce transport variance. In contrast, our work investigates \textbf{temporal coupling} on the multimodal product manifold that is orthogonal to the sample space coupling, aiming to minimize the transport cost along the entire generative trajectory.

\section{Methodology}

We consider the problem of \emph{multimodal generative modeling} for biomolecules, where heterogeneous modalities such as discrete sequences and continuous 3D structures must be generated in a coordinated manner.
A key difficulty lies in the fact that different modalities exhibit fundamentally different noise sensitivities and convergence rates during generative training.
As a result, commonly used coupling strategies either
(i) enforce a rigid synchronous evolution that misaligns modality-specific progress, or
(ii) rely on random decoupling that introduces high training variance and destabilizes learning.

Our core goal is therefore to {learn a modality-aware temporal coupling} that aligns the generative trajectory with the intrinsic difficulty of each modality, improving both training stability and cross-modal consistency.
To this end, we propose \ours, a principled framework that formulates multimodal generation as a \emph{Temporal Optimal Transport} (TOT) problem over the time domain.
This formulation enables us to treat the coupling between modalities as a geometric object to be optimized, rather than a fixed heuristic.


\paragraph{Overview.}
We first lay the groundwork from the perspective of Optimal Transport (OT) in the temporal domain (Sec. \ref{subsec:temporal_ot}). We then describe our proposed bi-level optimization objective to effectively minimize the transport cost (Sec. \ref{subsec:bilevel}) and the Bayesian optimization strategy for practically solving the geodesic coupling\footnote{Strictly speaking, the learned path is a minimum cost path over a cost surface learned by the Gaussian Process (GP) surrogate. We use the term ``geodesic'' in a generalized sense to represent the optimal trajectory minimizing the cumulative transport energy on the temporal product manifold.} optimization (Sec. \ref{subsec:bo}).



\subsection{Temporal Optimal Transport on Product Manifolds} \label{subsec:temporal_ot}

\paragraph{Notation.}
Multimodal biomolecular generation operates on a product manifold of heterogeneous modalities, such as continuous structure expressed by Cartesian coordinates \(\vr \in \mathbb{R}^{N \times 3}\) and discrete sequence \(\vh \in \mathbb{R}^{N \times K}\). (e.g., types of residues or atoms), with $N$ the dimension of atoms or residues and \(K\) the dimension of node feature space. 

\paragraph{Generative Process.}
The generative goal is to learn a time-dependent vector field \(v_t\) that transports a factorized prior \(\pi_0 = p(\vr) \otimes p(\vh)\) to the joint data distribution \(\pi_1 = p_{\text{data}}(\vr, \vh)\).
A fundamental challenge is that these modalities exhibit distinct noise sensitivities and convergence rates.
To perform a unified analysis of different training objectives, we introduce the definition of \textit{temporal coupling} as a monotonic time-warping function $\gamma \in \Gamma$, mapping a global progress variable $\tau\in [0,1]$ to local schedules $\boldsymbol{\gamma}(\tau) = [t_r(\tau), t_h(\tau)]$ for each modality,
where $t_r$ and $t_h$ represent the noise levels for structure and sequence, respectively. This allows the joint probability distribution to factorize dynamically along the noised trajectory $p_\tau^\gamma$:
\begin{equation}
    p_\tau^{\text{joint}}(\vr, \vh; \boldsymbol{\gamma}) =  p_{t_r(\tau)}(\vr) \cdot p_{t_h(\tau)}(\vh).
\end{equation}
Here, $\boldsymbol{\gamma}$ acts as a dynamic coupling, capturing the temporal dependency structure.

However, traditional methods have extended the factorization to the learned generative path, implicitly imposing an \textit{isotropic synchronous coupling}, forcing all modalities to to update at a uniform time spacing w.r.t. the marginal noise schedules.
Recently, advanced methods introduce a \textit{random coupling} that samples the timestep for each modality independently from $[0,1]^2$ during training \citep{pmlr-v235-campbell24a, wang2025dplm2, qiupiloting}, enriching the supervision signal and mitigating the potential exposure bias of implicit synchronous timesteps. However, the effectiveness of different coupling methods remain under-explored.

To understand how coupling strategies affect the generative process, we visualize the temporal dynamics of different methods in structure-based drug design in Fig.~\ref{fig:main}C, showing the evolution of structure mean squared error (MSE, Top) and sequence identity (Bottom) throughout the generation process (0→1), illustrating the high variance brought by random coupling in multimodal biomolecule generation:
For diffusion-based TargetDiff \citep{guan_3d_2023},
under an isotropic (synchronous) coupling, the structure modality remains inaccurate for a long portion of the trajectory, whereas an asynchronous coupling (TargetDiff*) that reallocates time budget can reduce structural error much earlier, and enjoys a higher validity rate.
For BFN-based MolCRAFT \citep{qu2024molcraft} and MolPilot \citep{qiupiloting}, the random coupling disrupts the convergence of structure generation, and prioritizes sequences.

This motivates treating the temporal coupling as a first-class design variable: instead of fixing a global time
for all modalities, we opt to learn a complexity-aware schedule $\gamma(\tau)=(t_r(\tau),t_h(\tau))$ to improve cross-modal consistency throughout the generation process.

To quantify the efficiency of a coupling, its total transport cost $\mathcal{E}(\boldsymbol{\gamma})$ is measured by the cumulative path energy required to transport the prior to the data distribution:
\begin{equation}
    \mathcal{E}(\gamma) = 
\mathbb{E}_{\rvx\sim\mathcal{D}}\Bigl[
  \int_{0}^{1}\mathcal{L}_{\mathrm{MSE}}(\rvx,\gamma(t))\,\mathrm{d}t
\Bigr]
    \label{eq:transport_cost}
\end{equation}


\begin{proposition}
    
Minimizing \(E(\gamma)\) over admissible schedules \(\gamma\)
can be interpreted as a \emph{constrained transport problem}
in the temporal domain.
Specifically, each schedule \(\gamma:[0,1]\to[0,1]^2\)
induces a coupling measure
\(\pi_\gamma := \gamma_\#\lambda \in \mathcal{P}([0,1]^2)\),
supported on a monotonic curve, and the transport cost admits
the representation
\[
\mathcal{E}(\gamma)
=
\int_{[0,1]^2} c(t_r,t_h)\, d\pi_\gamma(t_r,t_h),
\]
where \(c(t_r,t_h):=\mathbb{E}_{x\sim \mathcal D}
\big[L_{\mathrm{MSE}}(x,(t_r,t_h))\big]\).
\end{proposition}


 This extends Optimal Transport from the application of sample coupling to the temporal domain, namely the temporal coordination of heterogeneous modalities as in Table~\ref{tab:coupling_generalization}.

We now establish the theoretical connection between this geometric cost and the training dynamics involving different couplings, where the coupling complexity dominates. 

\begin{proposition}
\label{prop:vlb_transport}
Let $\gamma$ be a temporal coupling schedule. Let $\rvu^\gamma_t$ denote the \emph{stochastic} regression target at time $t$ induced by the coupling (where $\rvu^\gamma_t$ depends on $\rvx_0, \rvx_1$). Let $u^\gamma(\rvx_t,t) = \mathbb{E}[\rvu^\gamma_t \mid \rvx_t]$ denote the \emph{ideal} target vector field (the conditional mean).
Let $v_\theta(\rvx_t,t)$ be a learned vector field under $\mathcal{L}=\int_0^1
\mathbb{E}\Big[
\|v_\theta-u_t^\gamma\|^2
\Big]\,dt$.
Then 
the transport cost decomposes as
\begin{align}
\mathcal{E}(\gamma)
&=\int_0^1
\Bigl(
\underbrace{\|v_\theta(\rvx_t,t)-u^\gamma(\rvx_t,t)\|^2}_{\text{Bias}}
+
\underbrace{\Var(\rvu^\gamma_t\mid \rvx_t)}_{\text{Variance}}
\Bigr)\,dt.
\end{align}\label{eq:decomp}
\end{proposition}

Intuitively, $\mathcal{E}(\gamma)$ measures the difficulty of the generative task under schedule $\gamma$. Therefore, it is desirable to find the geodesic coupling $\gamma^* = \arg\min_{\gamma \in \Gamma} \mathcal{E}(\gamma)$. 

However, unlike classical optimal transport where the cost function is fixed, the transport cost in our setting depends on the learned vector field itself.
This creates a circular dependency: the optimal coupling depends on the generative dynamics, while the dynamics are trained under a given coupling.
This motivates the bilevel optimization strategy introduced next, which serves as a tractable proxy for searching low-energy geodesic couplings.

\begin{table}[t]
\centering
\caption{\textbf{A Unified Transport Perspective.} We unify spatial sample pairing and temporal schedule alignment under the classic coupling formulation: minimizing the expected cost $\mathbb{E}_{\pi}[c]$ over a valid coupling plan $\pi$. Our method effectively defines the ``cost'' as the generative loss and the ``coupling'' as the time schedule.}
\label{tab:coupling_generalization}
\resizebox{\columnwidth}{!}{%
\begin{tabular}{@{}l|c|c|c@{}}
\toprule
\textbf{Perspective} & \textbf{Transport Cost} $c$ & \textbf{Coupling} $\pi$ & \textbf{Optimal Solution} $\pi^*$ \\ \midrule
\text{Sample Coupling} & $\|x_1 - x_0\|^2$ & $\pi(x_0, x_1)$ & OT Map from Sinkhorn \\
\text{Temporal Coupling} & ${\mathcal{L}_{\rm MSE}}(\rvx, t_h, t_r)$ & ${\pi}_\gamma(t_h,t_r)$ & $\pi_{\gamma^*}$ from \ours \\ \bottomrule
\end{tabular}%
}
\end{table}

\subsection{Bi-level Optimization for Geodesic Coupling}\label{subsec:bilevel}

Solving the TOT problem presents a unique challenge distinct from classical OT where the transport cost (e.g., Euclidean distance) is fixed \citep{villani2008optimal}. In our temporal domain, it becomes the kinetic energy of a neural network $v_\theta$, which is itself learned from the data distribution induced by $\gamma$.

To break this dependency cycle, we propose a bilevel optimization strategy that acts as a tractable proxy for geodesic search.
 Our goal is to find a trajectory $\gamma \in \Gamma$ that is optimized for the flow model to learn, effectively minimizing the transport energy required to map the prior to the data:
\begin{align}
\min_{\gamma\in\Gamma}\;& \mathcal{J}(\gamma) = 
\mathbb{E}_{x\sim\mathcal{D}}\Bigl[
  \int_{0}^{1}\mathcal{L}_{\theta^*}(\rvx,\gamma(t))\,\mathrm{d}t
\Bigr] \label{eq:outer_loop} \\
\text{s.t.}\;&
\theta^* \;=\;\arg\min_{\theta}\mathcal{L}_{\mathrm{MSE}}(\theta,\gamma)\,, \label{eq:inner_loop}
\end{align}
where $\mathcal{L}_{\mathrm{MSE}}$ denotes the standard training objective conditioned on the path $\gamma$. 
In the inner loop (Eq. \ref{eq:inner_loop}), we minimize the bias term by training $\theta$ via maximum likelihood estimation to approximate the vector field induced by a fixed coupling $\gamma$. In the {outer loop} (Eq. \ref{eq:outer_loop}), we minimize the a proxy objective $\mathcal{J}(\gamma)$ with the learned vector field $\theta^*$. Intuitively, this steers the generative process to avoid the high energy barriers on the loss landscape, and seeks to capture the true geodesic of the product manifold.

We claim that when the regression bias is small for a fixed coupling
$\gamma$, the conditional variance term becomes the dominant
contributor to the overall energy $\mathcal{E}(\theta,\gamma)$, and the bilevel programming enables effective optimization:
\begin{proposition}
Let $\mathcal{L}(\theta^*, \gamma)$ be the converged loss for a fixed trajectory $\gamma$. The geodesic coupling $\gamma^*$ satisfies:
\begin{equation}
    \gamma^* = \arg\min_{\gamma \in \Gamma} \mathbb{E}_{t, \rvx} \left[ \text{Var}(u_t^\gamma | \rvx_t) \right].
\end{equation}
\end{proposition}

The above analysis motivates temporal couplings that reduce
intrinsic uncertainty along the transport path, which
empirically leads to more stable training and improved
sample quality metrics.
Intuitively, our bi-level optimization decouples the \textit{learning of the dynamics} (inner loop) from the \textit{discovery of the dependency structure} between modalities (outer loop). While isotropic coupling imposes an arbitrary diagonal line in the time domain (which may be high-curvature in the data manifold), and random coupling introduces high-variance noise, \ours~identifies the natural geodesic that creates the smoothest transformation between the prior and the data.




\subsection{Optimization via Gaussian Process Surrogates}
\label{subsec:bo}
In this section, we present our method to solve the outer loop optimization for the optimal geodesic coupling $\gamma$. 
Directly solving the bi-level optimization in Eq.~(\ref{eq:outer_loop}) via gradient descent is intractable for two primary reasons: 
(1) calculating hyper-gradients through the entire inner-loop training trajectory is computationally prohibitive; 
(2) the landscape of $\mathcal{J}(\gamma)$ with respect to biomolecular couplings often exhibits ``plateau phases'', i.e., regions where one modality halts while the other evolves, rendering gradient-based schedule learners ineffective.

To address these challenges, we exploit the structure of the problem by modeling the \textit{local cost landscape} over the temporal domain. Unlike black-box methods that search directly in the high-dimensional space of trajectories, we adopt a surrogate-based strategy via Gaussian Process (GP).
By learning a smooth approximation of the pointwise cost surface, we reduce the search for an optimal coupling to a geodesic path-finding problem on a learned manifold.
This design allows us to efficiently explore and compare candidate couplings using Bayesian optimization, while amortizing the cost of inner-loop training across iterations.

We define the transport cost density function $c(\mathbf{t}) := \mathbb{E}_{\rvx}[\mathcal{L}_{\text{MSE}}(\theta^*, \rvx, \mathbf{t})]$ for any time pair $\mathbf{t}=(t_r, t_h) \in [0,1]^2$.
We model this unknown surface using a Gaussian Process (GP) surrogate:
\begin{equation}
    c(\mathbf{t}) \sim \mathcal{GP}\left(\mu(\mathbf{t}), k(\mathbf{t}, \mathbf{t}') + \sigma_n^2 \delta_{\mathbf{t}, \mathbf{t}'}\right).
\end{equation}
To maintain computational efficiency and mitigate the impact of stale estimates from early training stages, we maintain a fixed-size {observation buffer} $\mathcal{B}$ with capacity $N_{\max}=1000$.
The GP is fitted dynamically on $\mathcal{B}$, which stores the most recent tuples of time coordinates and their corresponding losses, ensuring the surrogate landscape reflects the current capability of the generative model.

The complete procedure is summarized in Algorithm~\ref{alg:geocoupling}. 
It should be noted that this algorithm provides a stochastic and approximate mechanism for navigating this objective landscape, and performs a noisy outer-loop search over couplings using empirical loss estimates.
The proposed learned coupling lies between these extremes: it adaptively aligns denoising stages across modalities, reducing conditional variance while maintaining sufficient exploration, thereby leading to lower empirical training variance and a tighter variational bound in practice.


\section{Experiments}
\label{sec:experiments}

We conduct comprehensive evaluations on two primary tasks: structure-based drug design (SBDD) and unconditional protein sequence-structure co-design. 
A detailed description is provided in Appendix~\ref{app:setup}.

\subsection{Structure-based Drug Design}
\label{subsec:sbdd}

\begin{table*}[t]
\caption{Results for structure-based ligand design, comparing our \ours~with synchronous (Sync.) and random coupling methods. 
Top-2 results highlighted in \textbf{bold} and \ul{underlined}, respectively. 
$\dagger$: cited from \citet{qiupiloting}.
$\heartsuit$: calculated from released samples.
}
\centering
\label{tab:main}
\resizebox{\linewidth}{!}{%
\begin{tabular}{c l c cc cc cc c c c c c} 
\toprule
\multirow{2}{*}{Coupling} & \multirow{2}{*}{Methods} & PB-Valid & \multicolumn{2}{c}{Vina Score ($\downarrow$)} & \multicolumn{2}{c}{Vina Min ($\downarrow$)} & \multicolumn{2}{c}{Vina Dock ($\downarrow$)} & scRMSD & QED & SA & Div & Size \\
& & Avg. ($\uparrow$) & Avg. & Med. & Avg. & Med. & Avg. & Med. & \textless 2 \r{A} ($\uparrow$) & Avg. & Avg. & Avg. & Avg. \\ \midrule

&{Reference}$\dagger$ & 95.0\% & -6.36 & -6.46 & -6.71 & -6.49 & -7.45 & -7.26 & 34.0\% & 0.48 & 0.73 & - & 22.8 \\ \midrule

\multirow{7}{*}{{Sync.}} 
& AR$\dagger$             & 59.0\% & -5.75 & -5.64 & -6.18 & -5.88 & -6.75 & -6.62 & 36.5\% & 0.51 & 0.63 & 0.70 & 17.7 \\
& Pocket2Mol$\dagger$     & 72.3\% & -5.14 & -4.70 & -6.42 & -5.82 & -7.15 & -6.79 & 32.0\% & {0.57} & {0.76} & 0.69 & 17.7 \\
& TargetDiff$\dagger$     & 50.5\% & -5.47 & -6.30 & -6.64 & -6.83 & -7.80 & -7.91 & 37.1\% & 0.48 & 0.58 & {0.72} & 24.2 \\
& DiffSBDD$\dagger$       & 37.6\% & -1.44 & -4.91 & -4.52 & -5.84 & -7.14 & -7.30  & 18.7\% & 0.47 & 0.58 & {0.73} & 24.4 \\
& DecompDiff$\dagger$     & 71.7\% & -5.19 & -5.27 & -6.03 & -6.00 & -7.03 & -7.16 & 24.2\% & 0.51 & 0.66 & {0.73} & 21.2 \\
& MolCRAFT$\dagger$       & 84.6\% & -6.55 & -6.95 & -7.21 & -7.14 & -7.67 & -7.82 & \textbf{46.8\%} & 0.50 & 0.67 & {0.72} & 22.7 \\
& DrugFlow$\heartsuit$       & 79.6\% & -5.12 & -5.50 & -6.02 & -6.03 & -6.99 & -7.03 & 23.1\% & 0.52 & 0.72 & 0.71 & 21.4 \\ \midrule

\multirow{2}{*}{{Random}} 
& TargetDiff*$\dagger$ & {58.1\%} & {-6.46} & {-6.53} & {-7.04} & {-7.09} & \ul{-8.04} & \ul{-8.12} & {40.2\%}  & 0.49 & 0.59 & 0.71 & 25.7 \\
& MolPilot$\dagger$       & \textbf{95.9\%} & \ul{-6.88} & \ul{-7.03} & \ul{-7.23} & \ul{-7.27} & {-7.92} & {-7.92} & 41.1\% & {0.56} & 0.74 & 0.69 & 22.6 \\ \midrule

\rowcolor{gray!20}
{Learned}
& \ours          & \ul{94.3\%} & \textbf{-7.16} & \textbf{-7.43} & \textbf{-7.68} & \textbf{-7.62} & \textbf{-8.32} & \textbf{-8.43} & \ul{43.1\%} & 0.54 & {0.77} & 0.70 & 22.6 \\ \bottomrule
\end{tabular}
}
\end{table*}

\paragraph{Dataset.} 
Following standard SBDD protocols~\citep{luo_3d_2022, francoeur2020three}, we evaluate our method on the CrossDock dataset using the same split. The dataset contains 100,000 protein-ligand complexes for training and 100 complexes for test, where each provides a protein binding pocket and a reference ligand. The main experiment is conducted by sampling 100 molecules for each test protein.

\paragraph{Baselines.} 
We compare \ours~with representative methods from three major categories: 
(1) autoregressive models, including AR~\citep{luo_3d_2022} and Pocket2Mol~\citep{peng_pocket2mol_2022}; 
(2) diffusion-based models, including TargetDiff~\citep{guan_3d_2023}, DiffSBDD~\citep{schneuing2024structure}, and DecompDiff~\citep{guan_decompdiff_2023}; 
(3) Bayesian Flow Network (BFN)-based models, including MolCRAFT~\citep{qu2024molcraft} and MolPilot 
\citep{qiupiloting}; and (4) flow matching DrugFlow~\citep{schneuing_multi-domain_2025}.
All baseline results are taken from the released samples under the same evaluation protocol. 

\paragraph{Metrics.} 
We evaluate generated ligands from multiple perspectives. 
\textbf{Binding affinity} is assessed using AutoDock Vina \citep{eberhardt2021autodock}, reporting Vina Score, Vina Min, and Vina Dock (average and median; lower is better). 
\textbf{Pose validity} is measured by Posebusters validity ratio (PB-Valid), the percentage of generated molecules passing geometric and chemical validity checks \citep{buttenschoen2024posebusters}. 
\textbf{3D accuracy} is quantified by the percentage of samples with scRMSD below $2$\AA{} with respect to the reference ligand. 
\textbf{Molecular quality} is evaluated using QED and synthetic accessibility (SA), while \textbf{diversity} (Div) and average molecular size (Size) are reported to assess structural diversity and distributional alignment. 

\begin{figure*}[ht] 
    \centering
    
    \begin{minipage}[c]{0.66\textwidth}
        \centering
        \captionof{table}{Unconditional protein co-design results from 100 to 500 with sample size $N=100$. 
        Top-2 are \textbf{bold} / \ul{underlined}.
        $\diamondsuit$: baseline results calculated by us using the official code.
        }
        \label{tab:benchmark}
        
        \resizebox{\linewidth}{!}{%
            \begin{tabular}{lccccc}
    \toprule
    \multirow{2.5}{*}{\textbf{Method}} & \multicolumn{2}{c}{\textbf{Quality} ($\uparrow$)} & \multicolumn{2}{c}{\textbf{Diversity} ($\uparrow$)} & \multicolumn{1}{c}{\textbf{Novelty} ($\downarrow$)} \\
    \cmidrule(lr){2-3} \cmidrule(lr){4-5} \cmidrule(lr){6-6}
     & Co-design & pLDDT & 1 - Pairwise TM & FS Clusters & Max TM \\
    \midrule
    ProteinGenerator$\diamondsuit$          & 0.11 & 55.25 & 0.55 & 0.31 & 0.88 \\
    ProtPardelle$\diamondsuit$              & 0.38 & 63.46 & 0.42 & 0.08 & \ul{0.82} \\
    ProtPardelle-1c$\diamondsuit$           & 0.44 & 65.92 & 0.41 & 0.09 & \textbf{0.81} \\
    MultiFlow$\diamondsuit$                 & 0.72 & 79.39 & \ul{0.63} & 0.56 & 0.83 \\
    La-Proteina (no-tri)$\diamondsuit$      & 0.67 & 81.89 & \ul{0.63} & 0.64 & \textbf{0.81} \\
    La-Proteina (tri)$\diamondsuit$         & \ul{0.77} & \textbf{85.32} & 0.59 & 0.36 & 0.85 \\  
    DPLM2$\diamondsuit$                     & 0.31 & \ul{83.69} & \ul{0.63} & 0.49 & {0.96} \\  
    ESM3 (seq$\to$str)$\diamondsuit$        & 0.10 & 60.89 & 0.62 & {0.70} & 0.90 \\  
    ESM3 (str$\to$seq)$\diamondsuit$        & 0.05 & 61.10 & \ul{0.63} & \textbf{0.81} & 0.90 \\
    \midrule
    \rowcolor{gray!20}\ours   & \textbf{0.79} & 80.15 & \ul{0.63} & 0.48 & 0.83 \\ 
    \rowcolor{gray!20}\ours~(post-hoc) & 0.74 & 79.23 & \textbf{0.64} & \ul{0.73} & 0.83 \\ 
    \bottomrule
\end{tabular}
        }
    \end{minipage}%
    \hfill 
    \begin{minipage}[c]{0.3\textwidth}
        \centering
        \includegraphics[width=\linewidth]{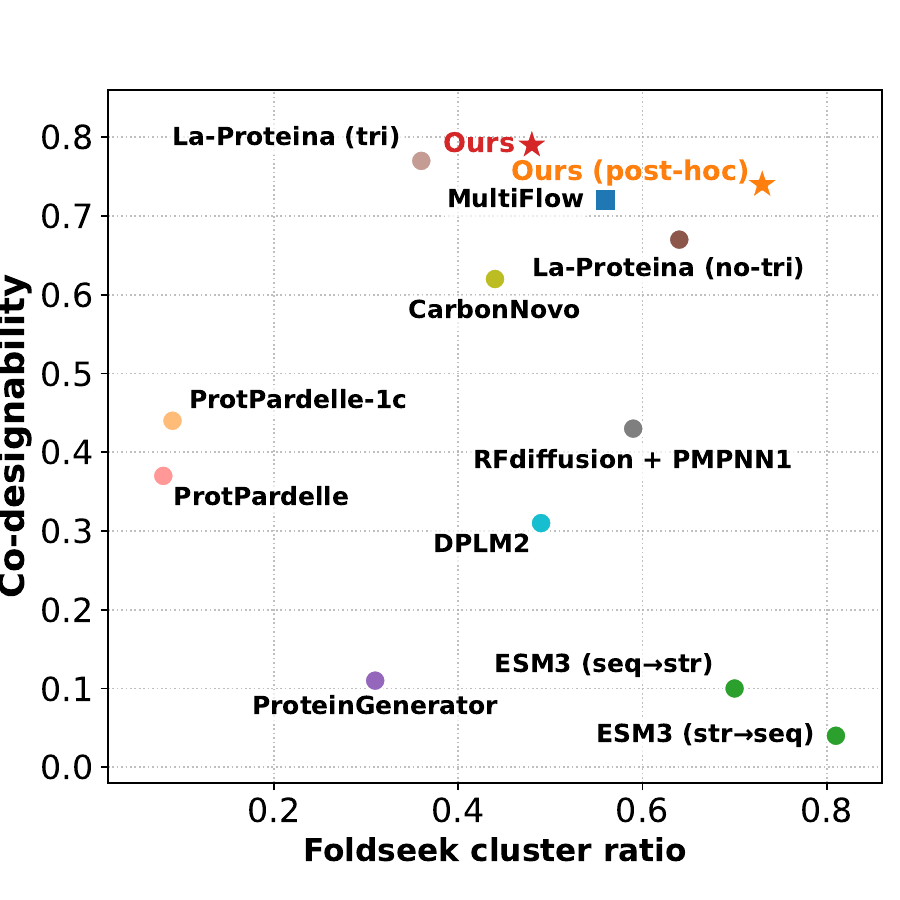}
        \caption{Co-designability vs. diversity trade-off (upper-right corner is better).}
        \label{fig:design_vs_div}
    \end{minipage}
\end{figure*}

\begin{figure*}[ht]
    \centering
    \includegraphics[width=\linewidth]{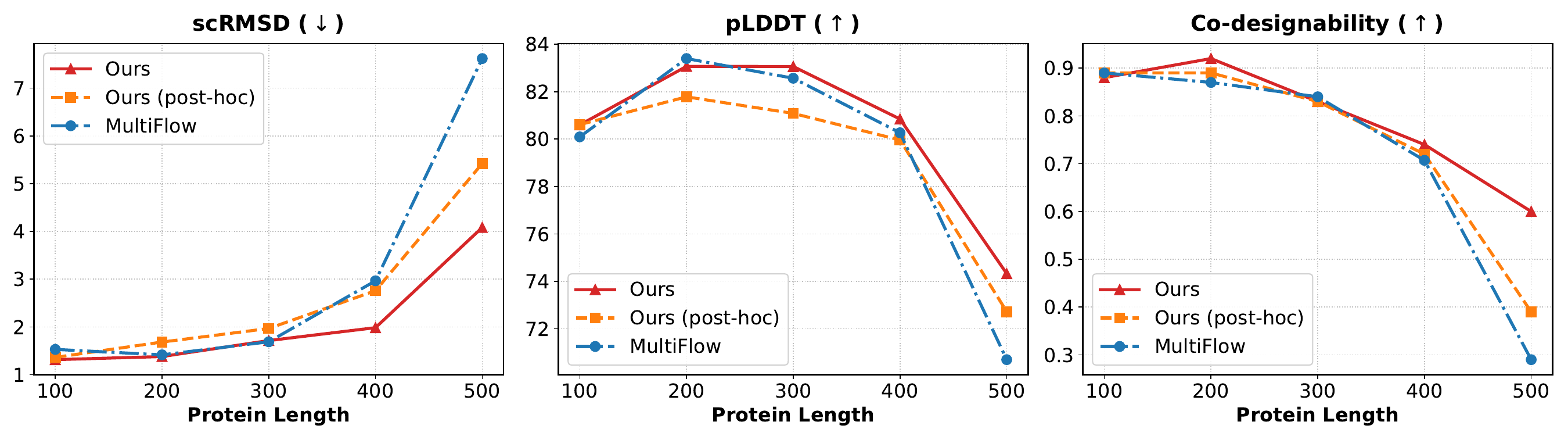}
\caption{{Performance Comparison Across Protein Lengths.} We evaluate scRMSD, pLDDT, and Co-designability across lengths 100-500. \ours~maintains robustness even in OOD regimes (Length $\ge$ 400), whereas baselines degrade significantly.}
    \label{fig:length_metrics}
\end{figure*}

\paragraph{Results.}
We summarize our findings below.

\textit{Learned coupling generally outperforms synchronous and independent baselines.}
Our \ours~that optimizes the coupling during training, consistently achieves state-of-the-art binding affinity across all metrics (Table~\ref{tab:main}), with an average Vina Score of $-7.11$ and Vina Dock of $-8.12$, outperforming both synchronous (Sync.) and random coupling methods.
Notably, \ours~attains systematically stronger protein-ligand interactions without sacrificing validity or drug-like properties, indicating that the learned generation paths better match the underlying data geometry.

\textit{Implicit synchronous coupling yields suboptimal sample quality.}
Early diffusion and BFN-based MolCRAFT, which rely on implicit synchronous coupling, generally achieve moderate PB-Valid rates and fall short in binding affinity.
In contrast, random coupling methods, represented by MolPilot and the TargetDiff* variant, substantially improve binding affinity over synchronous baselines, suggesting that entangling the generation trajectory with model updates restricts optimization efficiency.

\textit{Geodesic coupling leads to more efficient transport.}
MolPilot can be regarded as a special case of our bilevel optimization, where the outer loop is only conducted once after the training is converged. While superior to synchronous coupling methods, it is insufficient to fully align the generation trajectory with the training dynamics of the model, and requires 2$\times$ training steps to converge.
Fig.~\ref{fig:main}C illustrates the performance gap between our method and random coupling with post-hoc coupling selection. 
\ours~demonstrates significantly faster convergence compared to baselines. Notably, it achieves low structure error early in the process, which effectively guides the subsequent rapid recovery of sequence identity. This validates our claim that optimizing the temporal coupling facilitates cross-modal consistency, rather than compensating for suboptimal synchronous or random couplings.

\subsection{Protein Design}

\paragraph{Dataset.}
We follow \citet{pmlr-v235-campbell24a} and adopt the same training data, which consists of 18,694 monomeric protein structures obtained from \citet{yim2023se3} with lengths ranging from 60 to 384 residues, together with 4,179 PMPNN distillation samples.
For evaluation, we sample 100 proteins at each target length of 100, 200, 300, 400, and 500 residues, covering both in-distribution (ID) and length extrapolation (OOD) settings.

\paragraph{Baselines.}
We compare \ours~with representative protein co-design methods, categorized by their generative paradigms:
(1) {Diffusion Models}, including the sequence-space diffusion model ProteinGenerator \citep{lisanza2025multistate} and all-atom structure diffusion models ProtPardelle \citep{chu2024all} and ProtPardelle-1c \citep{lu2025conditional}; 
(2) {Flow Matching}, covering both sample-space methods like MultiFlow \citep{pmlr-v235-campbell24a} and partially latent-space methods such as La-Proteina \citep{geffner2026laproteina}; 
and (3) {Multimodal Token-based Models}, including the discrete diffusion model DPLM-2 \citep{wang2025dplm2} and the masked generative model ESM3 \citep{hayes2025simulating}, both of which jointly reason over sequence and structure tokens.

\paragraph{Metrics.}
We evaluate the generated proteins from complementary perspectives.
\textbf{Co-designability} is measured by the percentage of designed proteins exhibiting scRMSD \textless 2 \r{A} after ESMFold refolding.
\textbf{Diversity} is quantified in terms of structures by TM-score diversity, calculated as 1 - pairwise TM-score, and of sequences by FS cluster diversity, calculated as the number of Foldseek clusters.
\textbf{Novelty} is assessed structurally, calculated as the average TM-score for generated designable proteins with respect to its most similar counterpart in the PDB database using Foldseek.

\paragraph{Results.} We observe similar performance gains on the task of unconditional protein co-design in Table~\ref{tab:benchmark} (see Appendix~\ref{app-sec:eval_result} for more comprehensive results):

\textit{Our method consistently achieves strong designability while preserving competitive diversity and novelty.}
Training the generative model while iteratively performing bilevel optimization (\ours), yields the top co-designability (0.79) with reasonable structural novelty (0.83 average max TM-score).
As shown in Fig.~\ref{fig:design_vs_div}, our method occupies the optimal upper-right region of the trade-off landscape in terms of co-designability and Foldseek diversity.

\textit{Optimized coupling remains robust across both in-distribution and out-of-distribution protein lengths.}
Notably, Fig.~\ref{fig:length_metrics} shows that the advantage brought by \ours~becomes more pronounced in the OOD regime (length $\geq 400$).
While MultiFlow's performance drops significantly at length 500 (Co-designability $< 0.3$), \ours~maintains a high designable rate ($> 0.6$), showing that an optimized transport plan contributes to more robust performance. Notably, the longer-length proves to be more challenging for all models, with details in Appendix~\ref{app-sec:eval_result}.

\textit{The learned post-hoc coupling can be a plug-and-play enhancer for models trained by random coupling.}
After obtaining the optimized coupling, we apply it to the original MultiFlow checkpoint. As seen in Table~\ref{tab:benchmark} and Fig.~\ref{fig:design_vs_div}, the \ours~(post-hoc) variant significantly improves sequence diversity (FS Clusters 0.73) while maintaining competitive designability and even surpassing the previous state-of-the-art La-Proteina which is trained on the AFDB database \citep{FLEMING2025168967}.

\textit{Potentially universal structure-leading dynamics.}
We visualize the learned couplings in Fig.~\ref{fig:scheduler_comparison}, which reveals a shared property across both micro-scale (small molecules) and macro-scale (proteins) design tasks, i.e., \textbf{structure-leading geodesic}. The geometric structure evolves quicker in the early stages to establish a scaffold, while the sequence accelerates after the geometry is sufficiently formed.
This universal ``geometry-before-semantics'' discovery suggests that the dependencies between structure and sequence are governed by a common physical principle—geometry provides the necessary conditioning context for sequence decoding—which our bilevel optimization framework autonomously identifies and exploits.

\subsection{Ablation Study}

\begin{table}[t]
\centering
\caption{Ablation study on the contribution of bilevel optimization and the smoothing mechanism (EMA) with sample size 1000. ``Fixed Schedule'' refers to using a static scheduler derived a-priori, and ``w/o EMA'' denotes removing the trajectory smoothing.}
\label{tab:ablation_study}
\resizebox{\columnwidth}{!}{%
\begin{tabular}{@{}l|ccc|cc|cc@{}}
\toprule
\multirow{2}{*}{{Method}} & {Connected} & {QED} & {SA} & \multicolumn{2}{c|}{{Vina Score ($\downarrow$)}} & \multicolumn{2}{c}{{Vina Min ($\downarrow$)}} \\
 & ($\%, \uparrow$) & ($\uparrow$) & ($\uparrow$) & Mean & Median & Mean & Median  \\ \midrule
w/o EMA & 91.9 & \textbf{0.58} & 0.74 & -6.50 & -7.03 & -7.24 & -7.17 \\
Fixed $\gamma^*$ & 91.1& 0.54 & \textbf{0.76} & -6.97 & -7.22 & -7.45 & \textbf{-7.44} \\ \midrule
Ours & \textbf{93.5}& 0.54 & 0.75 & \textbf{-7.12} & \textbf{-7.24} & \textbf{-7.57} & -7.41 \\ \bottomrule
\end{tabular}%
}
\end{table}

To validate the effectiveness of our proposed bilevel optimization components, we conducted ablation studies focusing on the scheduler learning mechanism.

\paragraph{Effect of Bilevel Optimization.} 
On the task of structure-based drug design,
we compared \ours~with a fixed schedule variant, where the optimal coupling $\gamma^*$
  is determined a-priori and kept frozen throughout training. As shown in Table~\ref{tab:ablation_study}, the fixed schedule yields inferior performance in terms of binding affinity (Vina Score Mean -6.97 vs. -7.12) and connected rate (91.1\% vs. 93.5\%). This suggests that a static schedule fails to adapt to the evolving capability of the flow model, whereas the adaptive optimization without Exponential Moving Average (EMA) on the learned scheduler risks abrupt change of schedules and suffers in capturing the molecular geometry accurately. Our approach continuously and smoothly aligns the generative difficulty with the model's learning progress, capturing the intricate geometric dependencies and effectively balances chemical validity and binding affinity.
On the task of protein design, it can be seen that the structural plausibility is better captured by bilevel optimization, where the post-hoc variant falls short in codesignability, novelty as well as the ability to extrapolate beyond training distributions.

\begin{figure}[ht]
    \centering
    \includegraphics[width=\columnwidth]{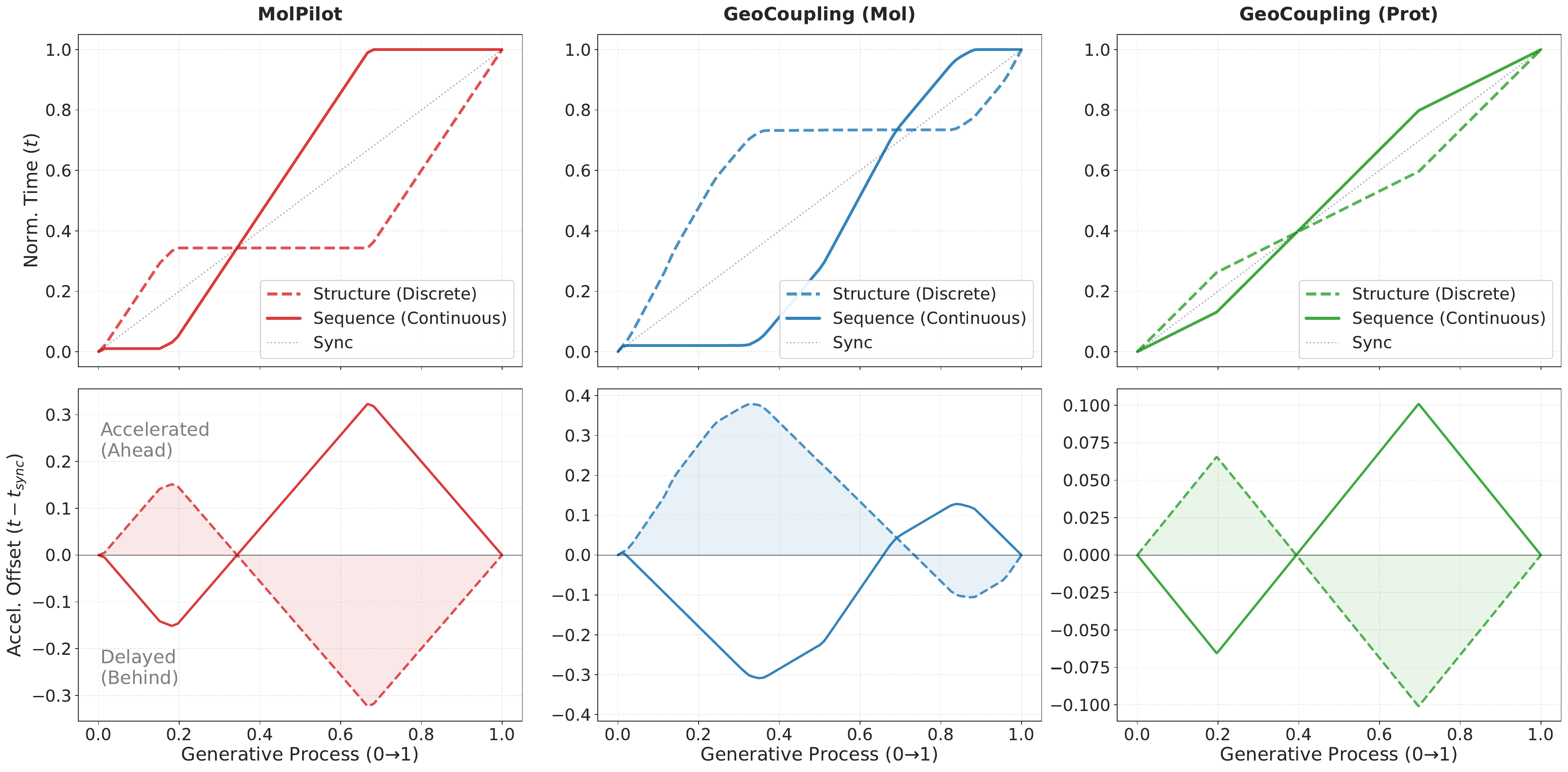}
    \caption{Comparison of the coupling learned by \ours~against baselines on different modalities, showing a \textit{structure-leading} geodesic for biomolecules, prioritizing geometric context before semantic sequence decoding.}
    \label{fig:scheduler_comparison}
\end{figure}

\paragraph{Computational Efficiency of Bayesian Optimization.} 
We measure the wall-clock time for optimization methods with and without Bayesian Optimization via Gaussian Process, i.e.,
a brute-force approach using discrete global search by constructing a dense transition grid, which requires over 1200 seconds per update cycle. The bottleneck lies in evaluating the transport cost across the dense grid (1213.6s), and the number of evaluation (NFEs) has a complexity of $O(N^K)$, with $K$ being the number of modalities and $N$ the grid resolution.
This means that directly solving for the optimal geodesic $\gamma^*$ on the discretized product manifold is computationally demanding. 
In contrast, \ours~leverages Bayesian Optimization to model the objective landscape with a Gaussian Process surrogate. This allows for efficient queries without exhaustive evaluation, reducing the optimization time to just 21.5 seconds—a \textbf{56$\times$ speedup}. This efficiency is critical, as it allows the outer-loop optimization to run frequently alongside the inner-loop model training without stalling the pipeline.

\section{Conclusion}
In this work, we identified the \textit{dynamic mismatch} between heterogeneous modalities as a critical bottleneck in multimodal biomolecular co-design. Existing methods relying on synchronous or random couplings often struggle with gradient conflicts and inefficient transport. To address this, we proposed {\ours}, a novel framework that formulates the generative process as a Temporal Optimal Transport (TOT) problem. By leveraging a bilevel Bayesian optimization strategy, \ours~learns the intrinsic geodesic coupling that dynamically aligns the generative difficulty of sequence and structure. We believe \ours~offers a generalizable paradigm for coordinating multimodal generative models in scientific domains.


\section*{Acknowledgements}
This work is supported by the Natural Science Foundation of China (Grant No. 62376133) and sponsored by Beijing Nova Program (20240484682).

\section*{Impact Statement}

This paper presents work whose goal is to advance generative modeling techniques within the domain of AI for Science and biomolecular discovery. Our proposed framework, \ours, aims to bridge the gap between geometric structure generation and sequence design, thereby improving the efficiency and quality of Structure-Based Drug Design (SBDD) and de novo protein engineering.

The primary societal benefit of this research lies in its potential to accelerate the development of novel therapeutics and functional biomolecules. By enabling the generation of more biophysically consistent and plausible molecules, our method contributes to reducing the high attrition rates and costs associated with early-stage drug discovery. This could facilitate faster delivery of treatments for complex diseases and support advancements in biotechnology, such as enzyme design for sustainability.

We acknowledge that advanced generative models in chemistry and biology carry inherent dual-use risks. The same capabilities that allow for the efficient exploration of therapeutic chemical space could theoretically be misused to design harmful compounds or toxins. While our work focuses on the fundamental optimization of multimodal generative processes rather than specific functional targeting, we emphasize the importance of responsible deployment. Future applications of such models should be accompanied by rigorous safety screening protocols, wet-lab validation, and adherence to ethical guidelines in biosecurity.



\bibliography{example_paper}

\begin{thebibliography}{46}
\providecommand{\natexlab}[1]{#1}
\providecommand{\url}[1]{\texttt{#1}}
\expandafter\ifx\csname urlstyle\endcsname\relax
  \providecommand{\doi}[1]{doi: #1}\else
  \providecommand{\doi}{doi: \begingroup \urlstyle{rm}\Url}\fi

\bibitem[Abramson et~al.(2024)Abramson, Adler, Dunger, Evans, Green, Pritzel, Ronneberger, Willmore, Ballard, Bambrick, et~al.]{abramson2024accurate}
Abramson, J., Adler, J., Dunger, J., Evans, R., Green, T., Pritzel, A., Ronneberger, O., Willmore, L., Ballard, A.~J., Bambrick, J., et~al.
\newblock Accurate structure prediction of biomolecular interactions with {AlphaFold} 3.
\newblock \emph{Nature}, 630\penalty0 (8016):\penalty0 493--500, 2024.

\bibitem[Banerjee et~al.(2005)Banerjee, Merugu, Dhillon, and Ghosh]{JMLR:v6:banerjee05b}
Banerjee, A., Merugu, S., Dhillon, I.~S., and Ghosh, J.
\newblock Clustering with {Bregman} divergences.
\newblock \emph{Journal of Machine Learning Research}, 6\penalty0 (58):\penalty0 1705--1749, 2005.
\newblock URL \url{http://jmlr.org/papers/v6/banerjee05b.html}.

\bibitem[Bose et~al.(2023)Bose, Akhound-Sadegh, Huguet, FATRAS, Rector-Brooks, Liu, Nica, Korablyov, Bronstein, and Tong]{bosefoldflow}
Bose, J., Akhound-Sadegh, T., Huguet, G., FATRAS, K., Rector-Brooks, J., Liu, C.-H., Nica, A.~C., Korablyov, M., Bronstein, M.~M., and Tong, A.
\newblock Se (3)-stochastic flow matching for protein backbone generation.
\newblock In \emph{The Twelfth International Conference on Learning Representations}, 2023.

\bibitem[Branden \& Tooze(2012)Branden and Tooze]{branden2012introduction}
Branden, C.~I. and Tooze, J.
\newblock \emph{Introduction to protein structure}.
\newblock Garland Science, 2012.

\bibitem[Buttenschoen et~al.(2024)Buttenschoen, Morris, and Deane]{buttenschoen2024posebusters}
Buttenschoen, M., Morris, G.~M., and Deane, C.~M.
\newblock {PoseBusters}: {AI}-based docking methods fail to generate physically valid poses or generalise to novel sequences.
\newblock \emph{Chemical Science}, 15\penalty0 (9):\penalty0 3130--3139, 2024.

\bibitem[Campbell et~al.(2024)Campbell, Yim, Barzilay, Rainforth, and Jaakkola]{pmlr-v235-campbell24a}
Campbell, A., Yim, J., Barzilay, R., Rainforth, T., and Jaakkola, T.
\newblock Generative flows on discrete state-spaces: Enabling multimodal flows with applications to protein co-design.
\newblock In \emph{Proceedings of the 41st International Conference on Machine Learning}, pp.\  5453--5512. PMLR, 2024.
\newblock URL \url{https://proceedings.mlr.press/v235/campbell24a.html}.

\bibitem[Chou(1995)]{chou1995novel}
Chou, K.-C.
\newblock A novel approach to predicting protein structural classes in a (20--1)-d amino acid composition space.
\newblock \emph{Proteins: Structure, Function, and Bioinformatics}, 21\penalty0 (4):\penalty0 319--344, 1995.

\bibitem[Chu et~al.(2024)Chu, Kim, Cheng, El~Nesr, Xu, Shuai, and Huang]{chu2024all}
Chu, A.~E., Kim, J., Cheng, L., El~Nesr, G., Xu, M., Shuai, R.~W., and Huang, P.-S.
\newblock An all-atom protein generative model.
\newblock \emph{Proceedings of the National Academy of Sciences}, 121\penalty0 (27):\penalty0 e2311500121, 2024.

\bibitem[Dauparas et~al.(2022)Dauparas, Anishchenko, Bennett, Bai, Ragotte, Milles, Wicky, Courbet, de~Haas, Bethel, et~al.]{dauparas2022robust}
Dauparas, J., Anishchenko, I., Bennett, N., Bai, H., Ragotte, R.~J., Milles, L.~F., Wicky, B.~I., Courbet, A., de~Haas, R.~J., Bethel, N., et~al.
\newblock Robust deep learning--based protein sequence design using proteinmpnn.
\newblock \emph{Science}, 378\penalty0 (6615):\penalty0 49--56, 2022.

\bibitem[Eberhardt et~al.(2021)Eberhardt, Santos-Martins, Tillack, and Forli]{eberhardt2021autodock}
Eberhardt, J., Santos-Martins, D., Tillack, A.~F., and Forli, S.
\newblock Autodock vina 1.2. 0: New docking methods, expanded force field, and python bindings.
\newblock \emph{Journal of chemical information and modeling}, 61\penalty0 (8):\penalty0 3891--3898, 2021.

\bibitem[Fleming et~al.(2025)Fleming, Magana, Nair, Tsenkov, Bertoni, Pidruchna, {Lima Afonso}, Midlik, Paramval, Žídek, Laydon, Kovalevskiy, Pan, Cheng, Žiga Avsec, Bycroft, Wong, Last, Mirdita, Steinegger, Kohli, Váradi, and Velankar]{FLEMING2025168967}
Fleming, J., Magana, P., Nair, S., Tsenkov, M., Bertoni, D., Pidruchna, I., {Lima Afonso}, M.~Q., Midlik, A., Paramval, U., Žídek, A., Laydon, A., Kovalevskiy, O., Pan, J., Cheng, J., Žiga Avsec, Bycroft, C., Wong, L.~H., Last, M., Mirdita, M., Steinegger, M., Kohli, P., Váradi, M., and Velankar, S.
\newblock {AlphaFold} protein structure database and 3d-beacons: New data and capabilities.
\newblock \emph{Journal of Molecular Biology}, 437\penalty0 (15):\penalty0 168967, 2025.
\newblock ISSN 0022-2836.
\newblock \doi{https://doi.org/10.1016/j.jmb.2025.168967}.
\newblock URL \url{https://www.sciencedirect.com/science/article/pii/S0022283625000336}.
\newblock Computation Resources for Molecular Biology.

\bibitem[Francoeur et~al.(2020)Francoeur, Masuda, Sunseri, Jia, Iovanisci, Snyder, and Koes]{francoeur2020three}
Francoeur, P.~G., Masuda, T., Sunseri, J., Jia, A., Iovanisci, R.~B., Snyder, I., and Koes, D.~R.
\newblock Three-dimensional convolutional neural networks and a cross-docked data set for structure-based drug design.
\newblock \emph{Journal of chemical information and modeling}, 60\penalty0 (9):\penalty0 4200--4215, 2020.

\bibitem[Geffner et~al.(2026)Geffner, Didi, Cao, Reidenbach, Zhang, Dallago, Kucukbenli, Kreis, and Vahdat]{geffner2026laproteina}
Geffner, T., Didi, K., Cao, Z., Reidenbach, D., Zhang, Z., Dallago, C., Kucukbenli, E., Kreis, K., and Vahdat, A.
\newblock La-proteina: Atomistic protein generation via partially latent flow matching.
\newblock In \emph{The Fourteenth International Conference on Learning Representations}, 2026.
\newblock URL \url{https://openreview.net/forum?id=RDerF20JYT}.

\bibitem[Graves et~al.(2023)Graves, Srivastava, Atkinson, and Gomez]{graves2023bayesian}
Graves, A., Srivastava, R.~K., Atkinson, T., and Gomez, F.
\newblock Bayesian flow networks.
\newblock \emph{arXiv preprint arXiv:2308.07037}, 2023.

\bibitem[Guan et~al.(2022)Guan, Qian, Peng, Su, Peng, and Ma]{guan_3d_2023}
Guan, J., Qian, W.~W., Peng, X., Su, Y., Peng, J., and Ma, J.
\newblock 3d equivariant diffusion for target-aware molecule generation and affinity prediction.
\newblock In \emph{The Eleventh International Conference on Learning Representations}, 2022.

\bibitem[Guan et~al.(2023)Guan, Zhou, Yang, Bao, Peng, Ma, Liu, Wang, and Gu]{guan_decompdiff_2023}
Guan, J., Zhou, X., Yang, Y., Bao, Y., Peng, J., Ma, J., Liu, Q., Wang, L., and Gu, Q.
\newblock {D}ecomp{D}iff: Diffusion models with decomposed priors for structure-based drug design.
\newblock In Krause, A., Brunskill, E., Cho, K., Engelhardt, B., Sabato, S., and Scarlett, J. (eds.), \emph{Proceedings of the 40th International Conference on Machine Learning}, volume 202 of \emph{Proceedings of Machine Learning Research}, pp.\  11827--11846. PMLR, 23--29 Jul 2023.
\newblock URL \url{https://proceedings.mlr.press/v202/guan23a.html}.

\bibitem[Hayes et~al.(2025)Hayes, Rao, Akin, Sofroniew, Oktay, Lin, Verkuil, Tran, Deaton, Wiggert, et~al.]{hayes2025simulating}
Hayes, T., Rao, R., Akin, H., Sofroniew, N.~J., Oktay, D., Lin, Z., Verkuil, R., Tran, V.~Q., Deaton, J., Wiggert, M., et~al.
\newblock Simulating 500 million years of evolution with a language model.
\newblock \emph{Science}, 387\penalty0 (6736):\penalty0 850--858, 2025.

\bibitem[Hsu et~al.(2022)Hsu, Verkuil, Liu, Lin, Hie, Sercu, Lerer, and Rives]{hsu_learning_2022}
Hsu, C., Verkuil, R., Liu, J., Lin, Z., Hie, B., Sercu, T., Lerer, A., and Rives, A.
\newblock Learning inverse folding from millions of predicted structures, April 2022.
\newblock URL \url{http://biorxiv.org/lookup/doi/10.1101/2022.04.10.487779}.

\bibitem[Huang et~al.(2016)Huang, Boyken, and Baker]{huang2016coming}
Huang, P.-S., Boyken, S.~E., and Baker, D.
\newblock The coming of age of de novo protein design.
\newblock \emph{Nature}, 537\penalty0 (7620):\penalty0 320--327, 2016.

\bibitem[Jumper et~al.(2021)Jumper, Evans, Pritzel, Green, Figurnov, Ronneberger, Tunyasuvunakool, Bates, {\v{Z}}{\'\i}dek, Potapenko, et~al.]{jumper2021highly}
Jumper, J., Evans, R., Pritzel, A., Green, T., Figurnov, M., Ronneberger, O., Tunyasuvunakool, K., Bates, R., {\v{Z}}{\'\i}dek, A., Potapenko, A., et~al.
\newblock Highly accurate protein structure prediction with {AlphaFold}.
\newblock \emph{Nature}, 596\penalty0 (7873):\penalty0 583--589, 2021.

\bibitem[Kingma \& Gao(2024)Kingma and Gao]{kingma2024understanding}
Kingma, D. and Gao, R.
\newblock Understanding diffusion objectives as the elbo with simple data augmentation.
\newblock \emph{Advances in Neural Information Processing Systems}, 36, 2024.

\bibitem[Lin et~al.(2023)Lin, Akin, Rao, Hie, Zhu, Lu, Smetanin, Verkuil, Kabeli, Shmueli, et~al.]{lin2023evolutionary}
Lin, Z., Akin, H., Rao, R., Hie, B., Zhu, Z., Lu, W., Smetanin, N., Verkuil, R., Kabeli, O., Shmueli, Y., et~al.
\newblock Evolutionary-scale prediction of atomic-level protein structure with a language model.
\newblock \emph{Science}, 379\penalty0 (6637):\penalty0 1123--1130, 2023.

\bibitem[Lipman et~al.(2022)Lipman, Chen, Ben-Hamu, Nickel, and Le]{lipman2022flow}
Lipman, Y., Chen, R.~T., Ben-Hamu, H., Nickel, M., and Le, M.
\newblock Flow matching for generative modeling.
\newblock In \emph{The Eleventh International Conference on Learning Representations}, 2022.

\bibitem[Lisanza et~al.(2025)Lisanza, Gershon, Tipps, Sims, Arnoldt, Hendel, Simma, Liu, Yase, Wu, et~al.]{lisanza2025multistate}
Lisanza, S.~L., Gershon, J.~M., Tipps, S.~W., Sims, J.~N., Arnoldt, L., Hendel, S.~J., Simma, M.~K., Liu, G., Yase, M., Wu, H., et~al.
\newblock Multistate and functional protein design using {RoseTTAFold} sequence space diffusion.
\newblock \emph{Nature biotechnology}, 43\penalty0 (8):\penalty0 1288--1298, 2025.

\bibitem[Liu et~al.(2022)Liu, Gong, and Liu]{liu2022flow}
Liu, X., Gong, C., and Liu, Q.
\newblock Flow straight and fast: Learning to generate and transfer data with rectified flow.
\newblock \emph{arXiv preprint arXiv:2209.03003}, 2022.

\bibitem[Lu et~al.(2026)Lu, Wang, Zhang, Gu, Jaitly, Susskind, and Bautista]{lu2026simpledesign}
Lu, J., Wang, Y., Zhang, Y., Gu, J., Jaitly, N., Susskind, J.~M., and Bautista, M.~{\'A}.
\newblock Simpledesign - a joint model for protein sequence and structure codesign, 2026.
\newblock URL \url{https://openreview.net/forum?id=ibbog3himK}.

\bibitem[Lu et~al.(2025)Lu, Shuai, Kouba, Li, Chen, Shirali, Kim, and Huang]{lu2025conditional}
Lu, T., Shuai, R., Kouba, P., Li, Z., Chen, Y., Shirali, A., Kim, J., and Huang, P.-S.
\newblock Conditional protein structure generation with protpardelle-1c.
\newblock \emph{bioRxiv}, 2025.

\bibitem[Luo et~al.(2021)Luo, Guan, Ma, and Peng]{luo_3d_2022}
Luo, S., Guan, J., Ma, J., and Peng, J.
\newblock A {3D} {Generative} {Model} for {Structure}-{Based} {Drug} {Design}.
\newblock \emph{Advances in Neural Information Processing Systems}, 34:\penalty0 6229--6239, 2021.
\newblock URL \url{http://arxiv.org/abs/2203.10446}.

\bibitem[Madani et~al.(2020{\natexlab{a}})Madani, McCann, Naik, Keskar, Anand, Eguchi, Huang, and Socher]{madani2020progen}
Madani, A., McCann, B., Naik, N., Keskar, N.~S., Anand, N., Eguchi, R.~R., Huang, P.-S., and Socher, R.
\newblock {ProGen}: Language modeling for protein generation.
\newblock \emph{arXiv preprint arXiv:2004.03497}, 2020{\natexlab{a}}.

\bibitem[Madani et~al.(2020{\natexlab{b}})Madani, McCann, Naik, Keskar, Anand, Eguchi, Huang, and Socher]{madani_progen_2020}
Madani, A., McCann, B., Naik, N., Keskar, N.~S., Anand, N., Eguchi, R.~R., Huang, P.-S., and Socher, R.
\newblock {ProGen}: {Language} {Modeling} for {Protein} {Generation}, March 2020{\natexlab{b}}.
\newblock URL \url{http://arxiv.org/abs/2004.03497}.
\newblock arXiv:2004.03497 [q-bio].

\bibitem[Nijkamp et~al.(2023)Nijkamp, Ruffolo, Weinstein, Naik, and Madani]{nijkamp2023progen2}
Nijkamp, E., Ruffolo, J.~A., Weinstein, E.~N., Naik, N., and Madani, A.
\newblock {ProGen}2: exploring the boundaries of protein language models.
\newblock \emph{Cell systems}, 14\penalty0 (11):\penalty0 968--978, 2023.

\bibitem[Peng et~al.(2022)Peng, Luo, Guan, Xie, Peng, and Ma]{peng_pocket2mol_2022}
Peng, X., Luo, S., Guan, J., Xie, Q., Peng, J., and Ma, J.
\newblock {P}ocket2{M}ol: Efficient molecular sampling based on 3{D} protein pockets.
\newblock In Chaudhuri, K., Jegelka, S., Song, L., Szepesvari, C., Niu, G., and Sabato, S. (eds.), \emph{Proceedings of the 39th International Conference on Machine Learning}, volume 162 of \emph{Proceedings of Machine Learning Research}, pp.\  17644--17655. PMLR, 17--23 Jul 2022.
\newblock URL \url{https://proceedings.mlr.press/v162/peng22b.html}.

\bibitem[Qiu et~al.(2025)Qiu, Song, Fan, Liu, Zhang, Zheng, Zhou, and Ma]{qiupiloting}
Qiu, K., Song, Y., Fan, Z., Liu, P., Zhang, Z., Zheng, M., Zhou, H., and Ma, W.-Y.
\newblock Piloting structure-based drug design via modality-specific optimal schedule.
\newblock In \emph{Forty-second International Conference on Machine Learning}, 2025.

\bibitem[Qu et~al.(2024)Qu, Qiu, Song, Gong, Han, Zheng, Zhou, and Ma]{qu2024molcraft}
Qu, Y., Qiu, K., Song, Y., Gong, J., Han, J., Zheng, M., Zhou, H., and Ma, W.-Y.
\newblock Mol{CRAFT}: Structure-based drug design in continuous parameter space.
\newblock In \emph{Forty-first International Conference on Machine Learning}, 2024.
\newblock URL \url{https://openreview.net/forum?id=KaAQu5rNU1}.

\bibitem[Ren et~al.(2024)Ren, Zhu, and Zhang]{pmlr-v235-ren24e}
Ren, M., Zhu, T., and Zhang, H.
\newblock {C}arbon{N}ovo: Joint design of protein structure and sequence using a unified energy-based model.
\newblock In Salakhutdinov, R., Kolter, Z., Heller, K., Weller, A., Oliver, N., Scarlett, J., and Berkenkamp, F. (eds.), \emph{Proceedings of the 41st International Conference on Machine Learning}, volume 235 of \emph{Proceedings of Machine Learning Research}, pp.\  42462--42483. PMLR, 21--27 Jul 2024.
\newblock URL \url{https://proceedings.mlr.press/v235/ren24e.html}.

\bibitem[Schneuing et~al.(2024)Schneuing, Harris, Du, Didi, Jamasb, Igashov, Du, Gomes, Blundell, Lio, et~al.]{schneuing2024structure}
Schneuing, A., Harris, C., Du, Y., Didi, K., Jamasb, A., Igashov, I., Du, W., Gomes, C., Blundell, T.~L., Lio, P., et~al.
\newblock {Structure-based {Drug} {Design} with {Equivariant} {Diffusion} {Models}}.
\newblock \emph{Nature Computational Science}, 4\penalty0 (12):\penalty0 899--909, 2024.

\bibitem[Schneuing et~al.(2025{\natexlab{a}})Schneuing, Igashov, Dobbelstein, Castiglione, Bronstein, and Correia]{schneuing_multi-domain_2025}
Schneuing, A., Igashov, I., Dobbelstein, A.~W., Castiglione, T., Bronstein, M., and Correia, B.
\newblock Multi-domain {Distribution} {Learning} for {De} {Novo} {Drug} {Design}, 2025{\natexlab{a}}.
\newblock URL \url{https://arxiv.org/abs/2508.17815}.
\newblock Version Number: 1.

\bibitem[Schneuing et~al.(2025{\natexlab{b}})Schneuing, Igashov, Dobbelstein, Castiglione, Bronstein, and Correia]{schneuing2025multidomain}
Schneuing, A., Igashov, I., Dobbelstein, A.~W., Castiglione, T., Bronstein, M.~M., and Correia, B.
\newblock Multi-domain distribution learning for de novo drug design.
\newblock In \emph{The Thirteenth International Conference on Learning Representations}, 2025{\natexlab{b}}.
\newblock URL \url{https://openreview.net/forum?id=g3VCIM94ke}.

\bibitem[Song et~al.(2024)Song, Gong, Xu, Cao, Lan, Ermon, Zhou, and Ma]{song2024equivariant}
Song, Y., Gong, J., Xu, M., Cao, Z., Lan, Y., Ermon, S., Zhou, H., and Ma, W.-Y.
\newblock Equivariant flow matching with hybrid probability transport for 3d molecule generation.
\newblock \emph{Advances in Neural Information Processing Systems}, 36, 2024.

\bibitem[Van~Kempen et~al.(2024)Van~Kempen, Kim, Tumescheit, Mirdita, Lee, Gilchrist, S{\"o}ding, and Steinegger]{van2024fast}
Van~Kempen, M., Kim, S.~S., Tumescheit, C., Mirdita, M., Lee, J., Gilchrist, C.~L., S{\"o}ding, J., and Steinegger, M.
\newblock Fast and accurate protein structure search with foldseek.
\newblock \emph{Nature biotechnology}, 42\penalty0 (2):\penalty0 243--246, 2024.

\bibitem[Villani et~al.(2008)]{villani2008optimal}
Villani, C. et~al.
\newblock \emph{Optimal transport: old and new}, volume 338.
\newblock Springer, 2008.

\bibitem[Wang et~al.(2024)Wang, Zheng, YE, Xue, Huang, and Gu]{wang2024diffusion}
Wang, X., Zheng, Z., YE, F., Xue, D., Huang, S., and Gu, Q.
\newblock Diffusion language models are versatile protein learners.
\newblock In \emph{Forty-first International Conference on Machine Learning}, 2024.
\newblock URL \url{https://openreview.net/forum?id=NUAbSFqyqb}.

\bibitem[Wang et~al.(2025)Wang, Zheng, YE, Xue, Huang, and Gu]{wang2025dplm2}
Wang, X., Zheng, Z., YE, F., Xue, D., Huang, S., and Gu, Q.
\newblock {DPLM}-2: A multimodal diffusion protein language model.
\newblock In \emph{The Thirteenth International Conference on Learning Representations}, 2025.
\newblock URL \url{https://openreview.net/forum?id=5z9GjHgerY}.

\bibitem[Watson et~al.(2023)Watson, Juergens, Bennett, Trippe, Yim, Eisenach, Ahern, Borst, Ragotte, Milles, et~al.]{watson2023novo}
Watson, J.~L., Juergens, D., Bennett, N.~R., Trippe, B.~L., Yim, J., Eisenach, H.~E., Ahern, W., Borst, A.~J., Ragotte, R.~J., Milles, L.~F., et~al.
\newblock De novo design of protein structure and function with rfdiffusion.
\newblock \emph{Nature}, 620\penalty0 (7976):\penalty0 1089--1100, 2023.

\bibitem[YE et~al.(2025)YE, Zheng, Xue, Shen, Wang, Ma, Wang, Wang, Zhou, and Gu]{ye2025proteinbench}
YE, F., Zheng, Z., Xue, D., Shen, Y., Wang, L., Ma, Y., Wang, Y., Wang, X., Zhou, X., and Gu, Q.
\newblock Proteinbench: A holistic evaluation of protein foundation models.
\newblock In \emph{The Thirteenth International Conference on Learning Representations}, 2025.
\newblock URL \url{https://openreview.net/forum?id=BksqWM8737}.

\bibitem[Yim et~al.(2023)Yim, Trippe, De~Bortoli, Mathieu, Doucet, Barzilay, and Jaakkola]{yim2023se3}
Yim, J., Trippe, B.~L., De~Bortoli, V., Mathieu, E., Doucet, A., Barzilay, R., and Jaakkola, T.
\newblock Se (3) diffusion model with application to protein backbone generation.
\newblock In \emph{Proceedings of the 40th International Conference on Machine Learning}, pp.\  40001--40039, 2023.

\end{thebibliography}
\bibliographystyle{icml2026}

\newpage
\appendix
\onecolumn

\section{Algorithm}
We summarize the bi-level optimization in Algorithm~\ref{alg:geocoupling}.

To maintain computational efficiency and mitigate the impact of stale estimates from early training stages, we maintain a fixed-size {Observation Buffer} $\mathcal{B}$ with capacity $N_{\max}=1000$.
The GP is fitted dynamically on $\mathcal{B}$, which stores the most recent tuples of time coordinates and their corresponding losses, ensuring the surrogate landscape reflects the current capability of the generative model.

Given the GP posterior mean $\mu_c(\mathbf{t})$, the optimal coupling $\gamma^*$ is the path that minimizes the line integral of the cost from $(0,0)$ to $(1,1)$.
This corresponds to finding the geodesic on a manifold with metric tensor $g_{ij} = \mu_c(\mathbf{t})^2 \delta_{ij}$, which is solved efficiently using Dynamic Programming (DP) on a discretized grid of $[0,1]^2$ \citep{qiupiloting}.
This approach allows us to find complex, non-monotonic dependencies (e.g., plateau phases) that simple parametric curves cannot represent.

Instead of treating an entire training run as a single scalar observation, we treat the mini-batch losses generated from the converged model as sparse, noisy observations of the surface $c(\mathbf{t})$, evaluated on a set of time pairs sampled along the current trajectory (Trust Region) and the global domain (Exploration).
These point-wise estimates are pushed into the buffer $\mathcal{B}$, displacing the oldest entries, and used to refine the landscape estimate for the next iteration.

We propose several techniques to stablize training, as a critical challenge in applying BO to this bilevel problem is \textit{off-policy evaluation}. 
If the model $\theta^*$ is trained solely on a thin deterministic manifold induced by a candidate $\gamma_k$, the evaluation of a perturbed candidate $\gamma'$ during the GP acquisition phase incurs high bias due to distribution shift.
To ensure the learned vector field $v_{\theta^*}$ remains valid for candidate couplings in the neighborhood of $\gamma_k$, we implement a {relaxed inner loop} with trajectory perturbation $\tilde{\gamma}(\tau) = \gamma(\tau) + \epsilon$ with $\epsilon \sim \mathcal{N}(0, \sigma^2)$, effectively creating a ``probability tube'' around the geodesic. This ensures $v_{\theta^*}$ covers a trust region around $\gamma_k$. Additionally,
to prevent the global landscape from deviating too far, we train with a mixture policy: with probability $p$, we sample from the current optimized coupling $\gamma_k$ (within its tube), and with probability $1-p$, we sample from a geometry-agnostic random coupling $\gamma_{\text{Rand}}$.

While Proposition~\ref{prop:vlb_transport} establishes a population-level bias–variance decomposition under idealized assumptions, our proposed solution should be viewed as an approximate stochastic estimator of the ordering induced by this decomposition, rather than a realization of the exact optimum.
It is designed to bias the search toward couplings with lower conditional variance according to the structure revealed by the proposition.

\begin{algorithm}[t]
\caption{Bi-level Optimization for Geodesic Coupling}
\label{alg:geocoupling}
\SetAlgoLined
\KwIn{BO iterations $K$, Inner steps $S$, Mixture rate $p$, Observation buffer $\mathcal{B}$ (FIFO Queue)}

Initialize GP surrogate with initial replay buffer $\mathcal{B}_0\leftarrow \varnothing$\;

\For{$k=1$ \KwTo $K$}{
    \tcc{1. Outer Loop: Propose Candidate}
    Fit GP to current buffer $\mathcal{B}_{k-1}$ to estimate cost surface $\mu_c(\mathbf{t})$\;
    
    Select candidate $\gamma_k = \arg\max_{\gamma} \alpha(\gamma; \mathcal{B}_{k-1})$ via dynamic programming\;
    
    \tcc{2. Inner Loop: Relaxed Training}
    \For{$s=1$ \KwTo $S$}{
        Sample batch $(\rvx, \ry)$\;
        Sample mode $m \sim \text{Bernoulli}(p)$\;
        \eIf{$m=1$ (Trust Region)}{
            Sample time $\tau \sim \mathcal{U}[0,1]$\;
            Construct target $t_r, t_h \sim \tilde{\gamma}(\tau) = \gamma_k(\tau) + \mathcal{N}(0, \sigma^2)$\;
        }{
            Sample independent times $t_r, t_h \sim \mathcal{U}[0,1]$\;
        }
        Update $\theta_k$ to minimize $\mathcal{L}_{\text{MSE}}(\theta_k, t_r, t_h)$\;
        
        Record the training loss $y_k \leftarrow \mathcal{L}_{\text{MSE}}(\theta_k, t_r, t_h)$\;
    Update history $\mathcal{B}_k \leftarrow \mathcal{B}_{k-1} \cup \{(t_r, t_h, y_k)\}$\;
    }

}
\end{algorithm}

\section{Proofs}

\subsection{Proof for Temporal Optimal Transport}
\begin{proof}
Plugging in the definition and by Fubini's theorem,
\begin{equation}
    \mathcal E(\gamma)
    \;=\;
    \mathbb E_{\rvx\sim\mathcal D}\Bigl[\int_0^1 \mathcal L_{\mathrm{MSE}}(\rvx,\gamma(\tau))\,d\tau\Bigr]
    \;=\;
    \int_0^1 \mathbb E_{\rvx\sim\mathcal D}\bigl[\mathcal L_{\mathrm{MSE}}(\rvx,\gamma(\tau))\bigr]\,d\tau
    \;=\;
    \int_0^1 c(\gamma(\tau))\,d\tau.
\end{equation}
Let $\pi_\gamma=\gamma_{\#}\lambda$ be the pushforward of $\lambda$ by $\gamma$. Then for any integrable
$\varphi:[0,1]^2\to\mathbb R$ we have the change-of-variables identity
\(
\int_0^1 \varphi(\gamma(\tau))\,d\tau = \int_{[0,1]^2} \varphi(t_r,t_h)\,d\pi_\gamma(t_r,t_h).
\)
Applying it to $\varphi=c$ gives
\begin{equation}
    \mathcal E(\gamma)=\int_{[0,1]^2} c(t_r,t_h)\,d\pi_\gamma(t_r,t_h).
\end{equation}
Finally, by construction the marginals of $\pi_\gamma$ are $(t_r)_{\#}\lambda=\mu_r$ and
$(t_h)_{\#}\lambda=\mu_h$, hence $\pi_\gamma\in \Pi(\mu_r,\mu_h)$. 
Therefore, minimizing $\mathcal E(\gamma)$ over deterministic monotonic schedules $\gamma$ corresponds to finding the optimal \emph{deterministic coupling}.
While strictly speaking a Monge map may not exist due to potential plateau phases (non-bijective segments), this problem is well-posed under the Kantorovich OT formulation where $\pi_\gamma$ is a singular measure supported on the curve $\gamma$.
\end{proof}

\subsection{Proof for Proposition 3.2}


\begin{proof}

For Eq.~\ref{eq:decomp}, let $\rvu^\gamma_t$ be the stochastic target associated with the path measure $Q^\gamma$. Write the squared-error term as
\begin{align}
\E\!\left[\|v_\theta(\rvx_t,t)-\rvu^\gamma_t\|^2\right]
&=
\E_{\rvx_t}\Bigl[
\E\!\left[\|v_\theta(\rvx_t,t)-\rvu^\gamma_t\|^2\mid \rvx_t\right]
\Bigr].
\end{align}
For a fixed $\rvx_t$, apply the conditional bias--variance identity with respect to the stochastic target $\rvu^\gamma_t$:
\begin{align}
\E\!\left[\|v_\theta-\rvu^\gamma_t\|^2\mid \rvx_t\right]
&=
\left\|v_\theta-\E[\rvu^\gamma_t\mid \rvx_t]\right\|^2
+\E\!\left[\left\|\rvu^\gamma_t-\E[\rvu^\gamma_t\mid \rvx_t]\right\|^2\mid \rvx_t\right] \\
&=
\underbrace{\left\|v_\theta-u^\gamma(\rvx_t,t)\right\|^2}_{\text{Bias}}
+\underbrace{\Var(\rvu^\gamma_t\mid \rvx_t)}_{\text{Variance}},
\end{align}
where we used the definition $u^\gamma(\rvx_t,t) := \E[\rvu^\gamma_t\mid \rvx_t]$.
Integrating over $t\in[0,1]$ and using the definition of $\mathcal{E}(\gamma)$ gives the desired decomposition.
\end{proof}

\subsection{Proof for Proposition 3.3}

\begin{proof}
Fix any trajectory $\gamma\in\Gamma$. Let $\rvx_t$ denote the intermediate state at time $t\in[0,1]$
(induced by $\gamma$), and let $\rvu^\gamma_t$ be the corresponding regression target.
Consider the inner-loop MSE objective
\begin{equation}
\label{eq:inner-mse}
\mathcal{L}_{\mathrm{MSE}}(\theta,\gamma)
\;:=\;
\int_0^1 \mathbb{E}\Bigl[\bigl\|v_\theta(\rvx_t,t)-\rvu^\gamma_t\bigr\|^2\Bigr]\,\mathrm{d}t.
\end{equation}
Define the conditional mean target
\begin{equation}
u^\gamma(\rvx_t,t) \;:=\; \mathbb{E}\bigl[\rvu^\gamma_t \mid \rvx_t\bigr].
\end{equation}

By the bias-variance decomposition proved above, for each $t$,
\begin{equation}
\label{eq:bv-given}
\mathbb{E}\Bigl[\bigl\|v_\theta(\rvx_t,t)-\rvu^\gamma_t\bigr\|^2\Bigr]
=
\mathbb{E}\Bigl[\bigl\|v_\theta(\rvx_t,t)-u^\gamma(\rvx_t,t)\bigr\|^2\Bigr]
+
\mathbb{E}\Bigl[\Var(\rvu^\gamma_t\mid \rvx_t)\Bigr].
\end{equation}
The second term in \eqref{eq:bv-given} is independent of $\theta$. Hence, minimizing the inner loop
over $\theta$ only affects the first term. Under the sufficient-capacity assumption and convergence
of the inner optimization, we can achieve zero approximation bias, i.e.,
\begin{equation}
\mathbb{E}\Bigl[\bigl\|v_{\theta^*}(\rvx_t,t)-u^\gamma(\rvx_t,t)\bigr\|^2\Bigr]=0
\quad\text{for a.e. }t,
\end{equation}
and therefore
\begin{equation}
\label{eq:converged-instant}
\mathbb{E}\Bigl[\bigl\|v_{\theta^*}(\rvx_t,t)-\rvu^\gamma_t\bigr\|^2\Bigr]
=
\mathbb{E}\Bigl[\Var(\rvu^\gamma_t\mid \rvx_t)\Bigr].
\end{equation}
Integrating \eqref{eq:converged-instant} over $t\in[0,1]$ yields the converged inner-loop loss
\begin{equation}
\label{eq:converged-loss}
\mathcal{L}_{\mathrm{MSE}}(\theta^*,\gamma)
=
\int_0^1 \mathbb{E}_{\rvx_t}\Bigl[\Var(\rvu^\gamma_t\mid \rvx_t)\Bigr]\,\mathrm{d}t.
\end{equation}

Finally, the outer loop selects $\gamma$ to minimize the converged loss. Using \eqref{eq:converged-loss}, this is exactly
\begin{equation}
\gamma^*
=
\arg\min_{\gamma\in\Gamma}\mathcal{L}_{\mathrm{MSE}}(\theta^*,\gamma)
=
\arg\min_{\gamma\in\Gamma}\int_0^1 \mathbb{E}_{\rvx_t}\Bigl[\Var(\rvu^\gamma_t\mid \rvx_t)\Bigr]\,\mathrm{d}t,
\end{equation}
which proves the claim.

\end{proof}



\section{Detailed Formulation for Multimodal Generative Modeling}
\label{sec:problem}

\subsection{A Unified View for Multimodal Generative Models}
\label{subsec:unified_view}

Recall that multimodal biomolecular design operates on a product manifold \(\mathcal{M} = \mathcal{M}_r \times \mathcal{M}_h\), involving continuous coordinates \(\vr \in \mathbb{R}^{N \times 3}\) and discrete features \(\vh \in \mathbb{R}^{N \times K}\) (e.g., types of residues), with $N$ the dimension of atoms or residues and \(K\) the dimension of node feature space. 

The generative goal is to learn a time-dependent vector field \(v_t\) that transports a factorized prior \(\pi_0 = p(\vr) \otimes p(\vh)\) to the joint data distribution \(\pi_1 = p_{\text{data}}(\vr, \vh)\).
We unify the training objective for continuous-time Diffusion Models, Flow Matching, and Bayesian Flow Networks as regressing a target vector $U$ from a network input state $Z$ under a squared loss. 
We formalize this by first introducing a {local-time pair} $\rvt:=(t_r,t_h)\in[0,1]^2$, and the {joint interpolation state} is then defined as the concatenation of marginal interpolants $\rvx_{\rvt} := \bigl(\psi^{r}_{t_r},\ \psi^{h}_{t_h}\bigr)$. 
The corresponding {regression target} is the joint velocity field $U_{\rvt} := \bigl(\dot\psi^{r}_{t_r},\ \dot\psi^{h}_{t_h}\bigr)$,
where $\dot\psi^{r}_{t_r}:=\partial_{t_r}\psi^{r}_{t_r}$ and $\dot\psi^{h}_{t_h}:=\partial_{t_h}\psi^{h}_{t_h}$. The network input is defined as the tuple $Z_{\rvt} := (\rvx_{\rvt},\rvt)$.

\paragraph{The General MSE Objective.}
The core difference between existing methods lies in the coupling distribution $\pi(\rvt)$ over local times. The generic training objective is:
\begin{equation}
    \mathcal L(\theta;\pi) = \mathbb E_{\rvt\sim\pi}\, \mathbb E\bigl[ \|v_\theta(Z_{\rvt})-U_{\rvt}\|^2 \bigr].
    \label{eq:mse_objective}
\end{equation}
And the ideal global minimizer satisfies the regression identity:
\begin{equation}
    v_{\theta^*}(Z_t)=\mathbb E[U_t\mid Z_t].
    \label{eq:cond_mean}
\end{equation}

Here, we show that for non-autoregressive generation approaches such as diffusion models (DM), flow matching (FM) and Bayesian flow networks (BFN), their training objective can be written as
$\mathcal L(\theta;\pi)=\mathbb E_{\rvt\sim\pi}\,\mathbb E\bigl[\|v_\theta(Z_{\rvt})-U_{\rvt}\|^2\bigr]$
by specifying the time law $\pi(\rvt)$ and the pair $(Z_{\rvt},U_{\rvt})$.

\paragraph{Diffusion Models (DM).}
A standard Gaussian forward noising process is
\begin{equation}
    Z_{\rvt} \equiv X_{\rvt} = \alpha_{\rvt} X_0 + \sigma_{\rvt}\,\varepsilon,
    \qquad X_0\sim p_{\rm data},\ \varepsilon\sim\mathcal N(0,I),
\end{equation}
with $\rvt\sim\pi_{\rm diff}$ (often uniform in $[0,1]$, or uniform in log-SNR).
Three common parameterizations correspond to different targets $U_{\rvt}$, and the details can be referred to in \citet{kingma2024understanding}.

For \text{(i) $\epsilon$-prediction:}
\begin{align}
    \quad
    &v_\theta(Z_{\rvt}) \equiv \epsilon_\theta(X_{\rvt},\rvt),\qquad
    U_{\rvt}\equiv \varepsilon, \\
    &\Rightarrow\ 
    \mathcal L_{\rm diff}(\theta)
    =\mathbb E_{\rvt\sim\pi_{\rm diff}}\,
      \mathbb E_{X_0,\varepsilon}\bigl[\|\epsilon_\theta(X_{\rvt},\rvt)-\varepsilon\|^2\bigr].
\end{align}
For \text{(ii) score-prediction:}
\begin{align}
    \quad
    &v_\theta(Z_{\rvt}) \equiv s_\theta(X_{\rvt},\rvt),\qquad
    U_{\rvt}\equiv \nabla_{x}\log q_{\rvt}(x\mid X_0)\big|_{x=X_{\rvt}}
    = -\frac{1}{\sigma_{\rvt}}\varepsilon, \\
    &\Rightarrow\
    \mathcal L_{\rm score}(\theta)
    =\mathbb E_{\rvt\sim\pi_{\rm diff}}\,
      \mathbb E_{X_0,\varepsilon}\Bigl[\Bigl\|s_\theta(X_{\rvt},\rvt)+\frac{1}{\sigma_{\rvt}}\varepsilon\Bigr\|^2\Bigr].
\end{align}
For \text{(iii) $\nu$-prediction:}
\begin{align}
\quad
&v_\theta(Z_{\rvt}) \equiv v_\theta(X_{\rvt},\rvt), \\
&U_{\rvt}\equiv \nu_{\rvt}
:= \alpha_{\rvt}\,\varepsilon - \sigma_{\rvt}\,X_0, \\
&\Rightarrow\ 
\mathcal L_{\nu\text{-pred}}(\theta)
=\mathbb E_{\rvt\sim\pi_{\rm diff}}\,
  \mathbb E_{X_0,\varepsilon}\bigl[\|v_\theta(X_{\rvt},\rvt)-(\alpha_{\rvt}\varepsilon-\sigma_{\rvt}X_0)\|^2\bigr].
\end{align}

\paragraph{Flow Matching (FM).}
Let $p_0$ be a base distribution and $p_1=p_{\rm data}$. Choose a coupling
$\gamma(x_0,x_1)$ with $(X_0,X_1)\sim\gamma$, and sample $\rvt\sim\pi_{\rm FM}$
(typically $\mathcal{U}[0,1]$). Pick a path $\psi_t$ and define
\begin{equation}
    Z_{\rvt}\equiv z_{\rvt}=\psi_{\rvt}(X_0,X_1),
    \qquad
    U_{\rvt}\equiv u_{\rvt}=\partial_t\psi_t(X_0,X_1)\big|_{t=\rvt}.
\end{equation}
Then
\begin{equation}
    \mathcal L_{\rm FM}(\theta)
    =\mathbb E_{\rvt\sim\pi_{\rm FM}}\,
      \mathbb E_{(X_0,X_1)\sim\gamma}\bigl[\|v_\theta(z_{\rvt},\rvt)-u_{\rvt}\|^2\bigr].
\end{equation}
For the linear interpolation path $\psi_t(X_0,X_1)=(1-t)X_0+tX_1$,
\begin{equation}
    Z_{\rvt}=(1-\rvt)X_0+\rvt X_1,\qquad U_{\rvt}=X_1-X_0.
\end{equation}

\paragraph{Discrete Flow Model (DFM).}
For discrete flow models on finite state space introduced in \citet{pmlr-v235-campbell24a}, the intermediate-state distribution $X_t \sim p_{t\mid 1}(\cdot\mid X_1)$ defined by the conditional transition rates 
$R_t(x,t\mid x_1)$, we have
\begin{equation}
    Z_{\rvt}\equiv z_{\rvt}:=(X_{\rvt},\rvt),
    \qquad
    U_{\rvt}\equiv u_{\rvt}:=X_1.
\end{equation}
The model outputs a categorical distribution
\begin{equation}
    v_\theta(Z_t)\;\equiv\;Q_\theta(\cdot\mid Z_{\rvt}) = p^\theta_{1\mid \rvt}(\cdot\mid X_{\rvt}) \in \Delta^{N-1},
\end{equation}
which is trained with the cross-entropy loss:
\begin{equation}
\mathcal L_{\rm DFM}(\theta)
=
\mathbb E_{x_1\sim p_{\rm data}}
\mathbb E_{t\sim\mathcal U([0,1])}
\mathbb E_{x_t\sim p_{t|1}(\cdot\mid x_1)}
\Bigl[-\log p^\theta_{1|t}(x_1\mid x_t)\Bigr].
\end{equation}
Equivalently, writing the true conditional $p(\cdot\mid Z_{\rvt}) := p(U_{\rvt}=\cdot\mid Z_{\rvt})$,
\begin{equation}
\mathbb E\bigl[-\log Q_\theta(U_{\rvt}\mid Z_{\rvt})\bigr]
=
H\bigl(p(\cdot\mid Z_{\rvt})\bigr)
+
\mathrm{KL}\!\Bigl(p(\cdot\mid Z_{\rvt})\,\big\|\,Q_\theta(\cdot\mid Z_{\rvt})\Bigr),
\end{equation}
so minimizing $\mathcal L_{\rm DFM}$  corresponds to minimizing the KL-divergence.

We admit that this is generally not an MSE form, but the conclusion can be generalized to Bregman divergence \citep{JMLR:v6:banerjee05b}.
\begin{equation}
\mathbb E\Big[\mathrm{KL}\bigl(p(\cdot\mid Z)\,\|\,Q(\cdot\mid Z)\bigr)\Big]
=
\mathrm{KL}\bigl(p(\cdot\mid Z)\,\|\,\bar q(\cdot\mid Z)\bigr)
\;+\;
\mathbb E\Big[\mathrm{KL}\bigl(\bar q(\cdot\mid Z)\,\|\,Q(\cdot\mid Z)\bigr)\Big],
\end{equation}
where $p(\cdot\mid Z)\equiv p(U=\cdot\mid Z)$, and the ensemble mean is defined by $\bar q(\cdot\mid Z):=\mathbb E[Q(\cdot\mid Z)]$.
Taking $Z\equiv Z_{\rvt}$ and then averaging over $\rvt\sim\pi_{\rm DFM}$ yields a global
``bias--variance'' decomposition for the DFM cross-entropy risk, with the bias term measured by
$\mathrm{KL}(p\|\bar q)$ and the variance term by $\mathbb E[\mathrm{KL}(\bar q\|Q)]$.


\paragraph{Bayesian Flow Networks (BFN).}
BFN uses a \emph{belief state} that accumulates information over time. For
continuous-time Gaussian form that parameterizes time by a (monotone) precision
schedule $\lambda(t)\ge 0$. With $X\sim p_{\rm data}$ and $\varepsilon\sim\mathcal N(0,I)$,
the (posterior) belief mean can be written as
\begin{equation}
    Z_t \equiv \mu_t
    = \frac{\lambda(t)}{1+\lambda(t)}\,X
    + \frac{\sqrt{\lambda(t)}}{1+\lambda(t)}\,\varepsilon,
    \qquad \mathrm{Var}[X\mid \mu_t]=\frac{1}{1+\lambda(t)}I.
\end{equation}
A standard BFN training target is the clean datum:
\begin{equation}
    v_\theta(Z_t)\equiv \hat X_\theta(\mu_t,t),
    \qquad
    U_t \equiv X.
\end{equation}
The continuous-time BFN objective is typically weighted by the information rate
$\dot\lambda(t)$; under \eqref{eq:mse_objective} this can be absorbed into $\pi$ by
\begin{equation}
    \pi_{\rm BFN}(t)\ \propto\ \dot\lambda(t),
\end{equation}
yielding
\begin{equation}
    \mathcal L_{\rm BFN}(\theta)
    =\mathbb E_{\rvt\sim\pi_{\rm BFN}}\,
      \mathbb E_{X,\varepsilon}\bigl[\|\hat X_\theta(\mu_{\rvt},\rvt)-X\|^2\bigr].
\end{equation}

For discrete data, Let $\{ \rve_k\}_{k=1}^K$ be the standard basis of $\mathbb R^K$ for $K$ classes, and define the concatenated one-hot
embedding
\begin{equation}
    \rve_X \;:=\; \bigl(\rve_{X^{(1)}},\dots,\rve_{X^{(D)}}\bigr)\in\mathbb R^{DK}.
\end{equation}
Sample $t\sim \mathcal U(0,1)$ and a BFN {belief state}
$\vartheta_t \sim p_F(\vartheta\mid X,t)$ from the forward Bayesian flow.
We have
\begin{equation}
    Z_t \equiv z_t := (\vartheta_t,t),
    \qquad
    U_t \equiv u_t := \rve_X .
\end{equation}

The network predicts per-dimension categorical probabilities
$p^{(d)}_{o,\theta}(k\mid \vartheta_t;t)$ for $d\in\{1,\dots,D\}$, $k\in\{1,\dots,K\}$.
Define the predicted expected one-hot vector in each dimension by
\begin{equation}
    \hat{\rve}^{(d)}_\theta(\vartheta_t,t)
    \;:=\;
    \sum_{k=1}^K p^{(d)}_{o,\theta}(k\mid \vartheta_t;t)\,\rve_k
    \;\in\;\mathbb R^{K},
\end{equation}
and concatenate
\begin{equation}
    v_\theta(Z_t)
    \;\equiv\;
    \hat{\rve}_\theta(\vartheta_t,t)
    \;:=\;
    \bigl(\hat{\rve}^{(1)}_\theta(\vartheta_t,t),\dots,\hat{\rve}^{(D)}_\theta(\vartheta_t,t)\bigr)
    \;\in\;\mathbb R^{DK}.
\end{equation}

The BFN objective for discrete data can be written as
\begin{equation}
    \mathcal L_{\rm BFN}(\theta)
    \;=\;
    K\,\beta(1)\;
    \mathbb E_{X\sim p_{\rm data}}
    \mathbb E_{t\sim \mathcal U(0,1)}
    \mathbb E_{\vartheta_t\sim p_F(\cdot\mid X,t)}
    \Bigl[
        t\,\bigl\|U_t - v_\theta(Z_t)\bigr\|^2
    \Bigr].
\end{equation}

\subsection{Detailed Formulation of Temporal Couplings}


Standard approaches implicitly assume that the factorization of the prior extends to the generative trajectory \citep{guan_3d_2023, qu2024molcraft}. By sharing a single global progress variable $t \in [0,1]$, they enforce a \textit{Synchronous Coupling}:
\begin{definition}[Synchronous Temporal Coupling]
Let $\psi_t^r$ and $\psi_t^h$ denote the method-specific marginal trajectories induced by
the corresponding noise schedules $\sigma_r(t), \sigma_h(t)$ and endpoint distributions $\pi_0, \pi_1$.
The joint interpolation path under synchronous coupling is
\begin{equation}
    \rvx_t^{\text{sync}}
    =
    \bigl( \psi_t^r(\vr_0, \vr_1),\ \psi_t^h(\vh_0, \vh_1) \bigr),
    \qquad t \in [0,1].
\end{equation}
\end{definition}

This imposes a rigid topological constraint: the joint trajectory is confined to a one-dimensional curve
$\{\rvx_t^{\text{sync}}:t\in[0,1]\}$, thus the learned vector field is trained to satisfy
\begin{equation}
v_{\theta}(\rvx_t^{\text{sync}},t)
\;\approx\;
\frac{d}{dt}\rvx_t^{\text{sync}}
=
\bigl(\dot{\psi}_t^r,\ \dot{\psi}_t^h\bigr),
\end{equation}
thereby forcing both modalities to evolve synchronously and monotonically with respect to their marginal noise levels.

We identify that this synchronous coupling fails to model \textit{asynchronous dependencies} common for hierarchical biomolecular assembly. 
For example, protein backbone topology often constrains side-chain packing, implying a causal structure where the formation of secondary structure elements (backbone) often precedes the determination of sequence, and the sequence again affects specific packing of side-chains and the adjustment of overall structure.
However, by forcing synchronous convergence, the model is unable to explore these potentially low-energy regions of the fitness landscape, as the learned vector field has not been trained to approximate the gradients on off-diagonal interpolation states.

\begin{table*}[t]
\centering
\caption{Comparison between synchronous, random, and learned time couplings.}
\label{tab:coupling-notation}
\resizebox{0.9\linewidth}{!}{
\begin{tabular}{@{}llll@{}}
\toprule
Coupling type
& Coupling measure $\pi$ on $[0,1]^2$
& Deterministic map $\gamma(t)$ 
& Interpolation state $\rvx_{\rvt}$ \\ \midrule

Synchronous
& $\pi_{\mathrm{sync}} = (\gamma_{\mathrm{sync}})_{\#}\mathrm{U}[0,1]$
& $\gamma_{\mathrm{sync}}(t)=(t,t)$
& $\rvx_{t}^{\mathrm{sync}}=(r_{t},h_{t})$ \\

Random (independent)
& $\pi_{\mathrm{rand}}=\mathrm{U}([0,1]^2)$
& N/A (not a scalar path)
& sample $\rvt\sim\pi_{\mathrm{rand}}$, then $\rvx_{\rvt}=(r_{t_r},h_{t_h})$ \\

Learned (GeoCoupling)
& $\pi_{\gamma}=(\gamma)_{\#}\mathrm{U}[0,1]$
& $\gamma(t)=\bigl[t_r(t),t_h(t)\bigr]^\top$
& $\rvx_{t}^{\gamma}=(r_{t_r(t)},h_{t_h(t)})$ \\

\bottomrule
\end{tabular}
}
\end{table*}

To alleviate dynamic mismatch, recent works have proposed \textit{Random Coupling}, where modalities are corrupted by independent timesteps $\rvt = (t_r, t_h) \sim \mathcal{U}([0,1]^2)$ during training. It succeeds in covering the full product space and allows the model to see decoupled interpolation states (e.g., noisy structure with clean sequence). With an MSE objective typical in diffusion, flow matching and Bayesian flow network, the learned
vector field satisfies
\begin{equation}
v_{\theta^*}(\rvx_{\rvt},\rvt)
\;=\;
\mathbb{E}\!\left[
(\dot\psi^r_{t_r},\dot\psi^h_{t_h})\mid \rvx_{\rvt},t_r,t_h
\right],
\end{equation}
i.e., it regresses to the conditional mean target rather than matching a single deterministic path derivative.
Therefore, the structure modality update can condition on the sequence modality at arbitrary noise levels and vice versa, thereby unlocking \emph{asynchronous dependencies}.

However, it introduces a severe learnability bottleneck and a severe training-inference gap rooted in the variance of the regression target.
Under random coupling, denoising times $(t_r, t_h)$ are sampled independently and uniformly from $[0,1]^2$.
The resulting reference vector field, if assumed to be under the flow matching objective,
\[
u_t^{\mathrm{rand}}(x_t) = \mathbb{E}[x_1 - x_0 \mid x_t]
\]
corresponds to a mixture over all noise levels and modality states.

Such mixing leads to a high conditional variance, since samples with incompatible noise semantics are aggregated when estimating $u_t$.
Formally, the law of total variance implies
\[
\operatorname{Var}(u_t^{\mathrm{rand}} \mid x_t)
=
\mathbb{E}[\operatorname{Var}(u_t \mid x_t, t_r, t_h)]
+
\operatorname{Var}(\mathbb{E}[u_t \mid x_t, t_r, t_h]),
\]
where the second term is maximized when $(t_r,t_h)$ are drawn independently.
This explains the empirically observed instability and slower convergence under random coupling.

Using Eq.~\ref{eq:bv-given}, this decomposition reveals the fundamental dilemma for current coupling methods:
\begin{itemize}
    \item \textbf{High Target Variance (Random Coupling):} When training with fully decoupled timesteps, the conditioning state \(Z_t\) (a noisy structure and sequence) can arise from valid data via infinitely many combinations of \((\vr_0, \vr_1, \vh_0, \vh_1)\) and time pairs. This maximizes the entropy of the posterior distribution \(p(U_t | Z_t)\). The network, learning the conditional mean \(\mathbb{E}[U_t|Z_t]\), produces a ``washed-out'' vector field that averages over conflicting directions, failing to capture sharp modes.
    \item \textbf{High Model Bias (Synchronous Coupling):} Conversely, synchronous coupling reduces variance but imposes a high bias. The arbitrary lockstep trajectory likely deviates significantly from the true geodesic of the data manifold, resulting in a complex vector field with high curl that is difficult for the network \(v_\theta\), incurring high approximation error.
\end{itemize}

This analysis motivates the need for our \ours: a method that minimizes the \textit{total} transport cost. This requires finding a coupling \(\gamma\) that is sufficiently deterministic to minimize Target Variance, yet flexible enough to follow the data's intrinsic geodesic to minimize model bias, effectively finding the natural order and dynamically allocating the time budget for biomolecular generation.

\section{Implementation Details}
\subsection{Structure-based Drug Design}
\paragraph{Training and Inference}
We train the model using the Adam optimizer with a learning rate of $5 \times 10^{-4}$ and a batch size of 16. All experiments are conducted on a single NVIDIA A100 GPU (80GB). The training converges within 100K steps (approximately 12 hours), which reduces the training duration by nearly 50\% compared to previous baselines. To stabilize the learning process, we apply an Exponential Moving Average (EMA) with a decay of 0.9 for the coupling updates. For inference, we perform 100 sampling steps, utilizing the variance reduction sampling strategy proposed by \citet{qu2024molcraft}.



\paragraph{Hyperparameters for Network}

Following the configuration in \citet{qiupiloting}, we construct our network using $k$-Nearest Neighbor (kNN) graphs with $k=32$. The model consists of $N=9$ layers with a hidden dimension of $d=128$, employing 16 attention heads and a dropout rate of 0.1. The total parameter count is approximately 3.1 million.

\subsection{Protein Co-design}

\paragraph{Training and Inference}
We train the model on a cluster of 8 NVIDIA A100 GPUs for approximately one day. Optimization is performed using AdamW with a learning rate of $10^{-4}$. To maximize GPU utilization while fitting variable-length proteins, we employ dynamic batching with a maximum batch size of 100 and a \texttt{max\_num\_res\_squared} threshold of $800,000$.


\paragraph{Network Architecture}
Following the configuration in \citet{pmlr-v235-campbell24a}, we construct the same backbone that is adapted from \citet{yim2023se3}, including 8 Transformer blocks with 4 layers and a dropout rate of 0.2, Invariant Point Attention (IPA) with hidden dimension 16, and an added 3-layered MLP for amino acid type prediction over 21 tokens (20 amino acids + 1 mask token). The total parameter count is approximately 21.8 million.

\section{Experiment Setup}\label{app:setup}

\subsection{Structure-based Drug Design}

\paragraph{Baselines}
To provide a comprehensive evaluation, we benchmark our method against a diverse set of representative SBDD approaches, categorized by their underlying generative paradigms:

\begin{itemize}
    \item \textbf{Autoregressive Models:} These methods construct molecules sequentially. \textbf{AR} \citep{luo_3d_2022} generates atoms atom-by-atom conditioned on the protein pocket using voxel-based density predictions and MCMC sampling. \textbf{Pocket2Mol} \citep{peng_pocket2mol_2022} improves upon this by employing an E(3)-equivariant geometric neural network to capture frontier atoms, thereby enhancing sampling efficiency and chemical validity.

    \item \textbf{Diffusion Models:} This category includes score-based generative models. \textbf{DiffSBDD} \citep{schneuing2024structure} introduces an E(3)-equivariant diffusion framework that jointly denoises atom types and coordinates in continuous space. \textbf{TargetDiff} \citep{guan_3d_2023} utilizes a non-equivariant SE(3) diffusion process, decomposing the generation of continuous coordinates and discrete atom types. \textbf{DecompDiff} \citep{guan_decompdiff_2023} further incorporates structural priors by decomposing ligands into scaffolds and side chains (arms) to guide the diffusion process.

    \item \textbf{BFN:} \textbf{MolCRAFT} \citep{qu2024molcraft} adapts the Bayesian Flow Network (BFN) to the 3D molecular domain, employing variance reduction techniques to outperform diffusion baselines. \textbf{MolPilot} \citep{qiupiloting} advances this frontier by proposing a VLB-Optimal Scheduling (VOS) strategy, training through decoupled timesteps and searching for a joint schedule at test time. Building on the VOS strategy, TargetDiff* is a reimplimented version for decoupled training and test-time schedule derivation.

    \item \textbf{Flow Matching:} \textbf{DrugFlow} \citep{schneuing2025multidomain} integrates continuous flow matching with discrete Markov bridges, endowed with an uncertainty estimate and learnable number of atoms.
\end{itemize}

\paragraph{Metrics}
We assess the quality of generated ligands using a suite of metrics covering binding affinity, chemical properties, and structural consistency:

\begin{itemize}
    \item \textbf{Binding Affinity:} We utilize AutoDock Vina \citep{eberhardt2021autodock} to estimate binding strength. We report:
    (1) \textbf{Vina Score}: The raw affinity of the generated pose within the pocket;
    (2) \textbf{Vina Min}: The affinity after performing local energy minimization;
    (3) \textbf{Vina Dock}: The optimal affinity achieved after re-docking the generated molecule.
    
    \item \textbf{Molecular Quality \& Properties:} Computed via RDKit, we verify drug-likeness using the Quantitative Estimation of Drug-likeness (\textbf{QED}) and assess ease of synthesis with the Synthetic Accessibility (\textbf{SA}) score. 

    \item \textbf{Structural Consistency:} To verify the stability of the binding modes, we calculate the \textbf{self-consistency RMSD (scRMSD)}, which measures the deviation between the generated pose and its re-docked conformation (by Vina). We report the percentage of molecules with an scRMSD $< 2\text{\AA}$ compared to the redocked binding poses.
\end{itemize}

\subsection{Protein Co-design}

\paragraph{Baselines}
Following \citet{lu2026simpledesign}, we compare our method against a comprehensive set of state-of-the-art protein co-design models, categorized by their generative representations and mechanisms:

\begin{itemize}
    \item \textbf{Multimodal Language Models} that relies on structure tokenization, including \textbf{ESM3} \citep{hayes2025simulating} trained via masked generative modeling, and \textbf{DPLM-2} \citep{wang2025dplm2} that focuses on adapting a large-scale sequence-pretrained DPLM \citep{wang2024diffusion} for joint modeling. It is finetuned with discrete diffusion training objectives with a quantization-based tokenizer to convert continuous coordinates into discrete structural tokens. 

    \item \textbf{Sequence-Structure Diffusion:} \textbf{ProteinGenerator} \citep{lisanza2025multistate} utilizes a sequence-space diffusion process with folding performed by RoseTTAFold, allowing for the simultaneous generation of sequences and structures conditioned on desired functional attributes.
    
    \item \textbf{Hybrid Flow Matching:} \textbf{MultiFlow} \citep{pmlr-v235-campbell24a} combines continuous flow matching for SE(3) backbone structures with discrete flow models for amino acid sequences, enabling synchronized updates across both modalities during the generative trajectory.

    \item \textbf{All-Atom Latent Models:} \textbf{Protpardelle} \citep{chu2024all} introduces an all-atom diffusion model that manages the discrete-continuous nature of sidechains through a superposition state, collapsing them into specific conformations during sampling. \textbf{La-Proteina} \citep{geffner2026laproteina} proposes a partially latent flow matching framework, where the coarse backbone is modeled explicitly while atomistic details and sequences are captured via per-residue latent variables, effectively decoupling the complexity of side-chain packing.
\end{itemize}

We also include backbone-first methods to better position the co-design baselines in the wider context, such as \textbf{RFdiffusion} \citep{watson2023novo} integrated with ProteinMPNN \citep{dauparas2022robust}, and \textbf{CarbonNovo} \citep{pmlr-v235-ren24e} that improves this two-stage pipeline using an energy-based model with recycling, moving closer to a jointly co-design baseline.

\paragraph{Metrics}
Following standard evaluation protocols in co-design tasks \citep{yim2023se3, pmlr-v235-campbell24a}, we evaluate the generated proteins from complementary perspectives, focusing on designability, foldability, diversity and novelty:

\begin{itemize}
    \item \textbf{Co-designability:} This measures the self-consistency between the generated sequence and structure. We predict the structure of the generated sequence using ESMFold \citep{lin2023evolutionary} as the oracle folding model, and calculate the Root Mean Square Deviation (RMSD) between the generated backbone and the predicted structure. We report the self-consistency RMSD (\textbf{scRMSD}), specifically the success rate where scRMSD $< 2.0 \text{\AA}$, indicating that the generated sequence confidently folds into the designed structure.

    \item \textbf{Designability:} This measures whether a generated backbone can be realized by a compatible amino-acid sequence, independent of the sequence proposed by the generative model itself. Following established evaluations for backbone generation, we redesign each generated backbone using ProteinMPNN and then predict the structures of the redesigned sequences with ESMFold as the oracle folding model. We compute the  scRMSD between the original generated backbone and the ESMFold-predicted structure, and report the success rate where scRMSD $< 2.0 \text{\AA}$. In our co-design setting, we used a single ProteinMPNN redesign per backbone (\textbf{PMPNN-1}) to enable a fair comparison with co-design models, which typically generate only one sequence per structure. However, standard backbone-design benchmarks such as ProteinBench \citep{ye2025proteinbench} evaluate designability using eight ProteinMPNN samples per backbone and report the \textbf{best-of-8} success rate (\textbf{PMPNN-8}), counting a backbone as designable if at least one redesigned sequence folds back to the target structure, thus we additionally provide PMPNN-8 designability for better contextualization.

    \item \textbf{Foldability:} We report the average predicted Local Distance Difference Test (\textbf{pLDDT}) scores from ESMFold for the generated sequences. High pLDDT scores indicate that the sequence is predicted to fold into a well-defined, stable structure with high confidence.
    
    \item \textbf{Diversity:} We assess the diversity of the generated samples in both structural and sequence spaces. Structural diversity is measured by \textbf{TM-score diversity}, calculated as $1 - \text{average pairwise TM-score}$. Sequence-structure diversity is evaluated using \textbf{Foldseek cluster diversity}, defined as the total number of unique clusters identified by Foldseek \citep{van2024fast} within the generated set, ensuring the model does not collapse to a few modes.
    
    \item \textbf{Novelty:} To determine whether the model generates novel folds rather than memorizing training data, we calculate the average maximum \textbf{TM-score} between each designable generated protein and its nearest neighbor in the PDB database with a cutoff date 2025-12-16 (retrieved via Foldseek). A lower value indicates higher structural novelty.
    
\end{itemize}

\paragraph{Inference}
For all baselines, we run unconditional generation with $N=100$ samples for each target length $L\in\{100,200,300,400,500\}$, and use the released pretrained checkpoints from the corresponding official implementations.

\begin{itemize}
    \item \textbf{ESM3}\footnote{\url{https://github.com/evolutionaryscale/esm}} is evaluated with \texttt{esm3\_sm\_open\_v1} in two generation orders, sequence$\to$structure and structure$\to$sequence. Since the order and per-track temperature are meaningful choices for ESM3, we use $T=1.0$ for the first modality and $T=0.7$ for the second modality, with the repository-recommended number of sampling steps for sequence and structure tokens.

    \item \textbf{DPLM2}\footnote{\url{https://github.com/bytedance/dplm}} is run with the \texttt{airkingbd/dplm2\_650m} checkpoint and the released structure tokenizer. We use the DPLM2 sequence--structure co-generation mode with the recommended \texttt{annealing@2.0:0.1} strategy and 500 denoising iterations.

    \item \textbf{ProteinGenerator}\footnote{\url{https://github.com/RosettaCommons/protein_generator}} is run with the released sequence-diffusion checkpoint and the repository-recommended unconditional sampling configuration.

    \item \textbf{MultiFlow}\footnote{\url{https://github.com/jasonkyuyim/multiflow}} is evaluated using the official unconditional inference configuration, \texttt{inference\_unconditional}, together with the released unconditional model checkpoint.

    \item \textbf{Protpardelle}\footnote{\url{https://github.com/ProteinDesignLab/protpardelle}} is sampled with the released all-atom model and its default unconditional sampling configuration, \texttt{configs/uncond\_sampling.yml}.

    \item \textbf{Protpardelle-1c}\footnote{\url{https://github.com/ProteinDesignLab/protpardelle-1c}} is evaluated with the released Zenodo model bundle using the unconditional all-atom configuration \texttt{["cc91", "383", "sampling\_unconditional\_allatom\_s1"]}. We disable internal ProteinMPNN sampling for this variant so that the generated all-atom samples are evaluated as produced.

    \item \textbf{La-Proteina}\footnote{\url{https://github.com/NVIDIA-Digital-Bio/la-proteina}} is evaluated in both released unconditional variants, \texttt{LD1\_ucond\_notri\_512.ckpt} and \texttt{LD2\_ucond\_tri\_512.ckpt}, paired with the released autoencoder. We follow the repository's unconditional code-generation configuration and perform all designability metrics afterwards with our unified evaluation pipeline.

    \item \textbf{RFdiffusion + ProteinMPNN}\footnote{\url{https://github.com/RosettaCommons/RFdiffusion}} uses RFdiffusion for backbone generation and the official ProteinMPNN implementation\footnote{\url{https://github.com/dauparas/ProteinMPNN}} for inverse folding. RFdiffusion follows its default inference configuration. We report both the single-sequence ProteinMPNN setting (PMPNN-1) and the best-of-8 setting (PMPNN-8) for designability comparison.

    \item \textbf{CarbonNovo}\footnote{\url{https://github.com/CarbonMatrixLab/carbonnovo}} is run from the official repository with the released \texttt{carbonnovo.ckpt} and \texttt{igso3} parameter bundle, following the default inference configuration.
\end{itemize}

\begin{figure}[htbp]
    \centering
    \includegraphics[width=0.3\linewidth]{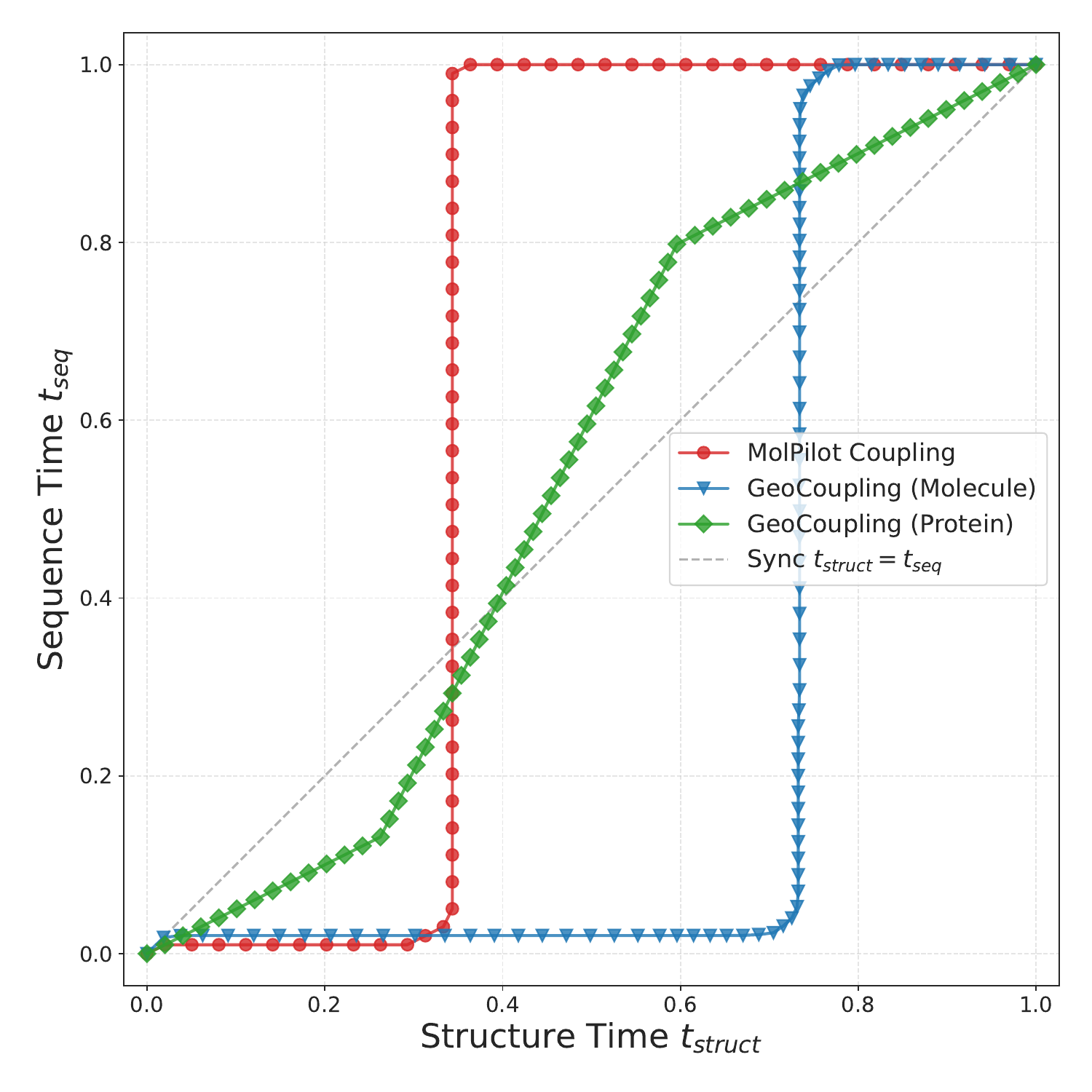}
    \caption{Visualization of the learned coupling against baselines.}
    \label{fig:more_comparison}
\end{figure}

\section{Extended Evaluation Results}\label{app-sec:eval_result}

\subsection{Protein Design}

\paragraph{Comprehensive Baseline Evaluation}
We additionally report redesign-based designability by ProteinMPNN to separate backbone quality from native joint consistency. 
Concretely, \textbf{PMPNN-1 Des.} redesigns each generated backbone with one ProteinMPNN sequence, whereas \textbf{PMPNN-8 Des.} reports the best-of-8 success rate, following standard backbone-generation benchmarks. 
As shown in Table~\ref{tab:appendix_full_protein}, \ours~balances strong joint co-design performance with high backbone designability under this structure-isolated evaluation, obtaining strong scores of $0.72$ (PMPNN-1) and $0.91$ (PMPNN-8), showing that its generated backbones remain highly realizable even under an external sequence redesign model. 

We also include \textbf{CarbonNovo} and \textbf{RFdiffusion+PMPNN} as additional structure-oriented baselines. CarbonNovo is competitive on redesignability ($0.75$ on PMPNN-8 and $0.47$ on PMPNN-1), while RFdiffusion performs slightly worse on both native and redesign-based metrics. 
The reason accounting for this performance degradation is attributed to the evaluated length range. Our benchmark setting averages over lengths $[100,200,300,400,500]$, whereas prior works often focus on shorter proteins up to length $250$. Since longer proteins are substantially more challenging, including lengths $300$--$500$ naturally lowers the overall designability average. Importantly, RFdiffusion remains strong on shorter lengths: when we restrict evaluation to $[100,200]$, our replicated RFdiffusion+PMPNN achieves $0.938 \pm 0.023$ PMPNN-8 designability. Therefore, the discrepancy mainly reflects the inclusion of more difficult long-length targets rather than a contradiction with earlier results.

\paragraph{Secondary Structure Analysis}
To test whether overall co-designability is inflated by a collapse toward mainly-$\alpha$ folds, we report in Table~\ref{tab:appendix_secondary_structure} to show that \ours{} generates proteins with similar secondary structure topology relative to MultiFlow and is not biased toward helices.

We further stratify generated proteins by secondary-structure class. Since CATH fold labels are unavailable in the unconditional setting, we follow \citet{chou1995novel} and group samples into mainly-$\alpha$ ($\alpha \geq 40\%, \beta \leq 5\%$), mainly-$\beta$ ($\alpha \leq 5\%, \beta \geq 40\%$), mixed $\alpha/\beta$ ($\alpha \geq 15\%, \beta \geq 15\%$), and others. 
Under the stratified analysis, \ours{} remains strongly designable across all classes. Compared with MultiFlow, \ours{} is consistently showing better co-designability with less variance, especially from $59\%$ to $83\%$ on others. These results suggest that the gains of \ours{} cannot be explained simply by a helical-collapse bias, but reflect more robust co-design performance across secondary-structure regimes.

\begin{table}[t]
\centering
\caption{Unconditional protein co-design results from 100 to 500 with sample size $N=100$.
        Top-2 mean values are \textbf{bold} / \ul{underlined}. Metrics are reported as mean $\pm$ sample std across three independent runs. We explicitly contextualize joint \textit{Co-designability} (native joint consistency) and structure-isolated \textit{Designability} using ProteinMPNN-1 and best-of-8 ProteinMPNN sequences for a more comprehensive comparison.}
\label{tab:appendix_full_protein}
\scriptsize
\resizebox{\linewidth}{!}{
\begin{tabular}{lccccccc}
    \toprule
    \multirow{2.5}{*}{\textbf{Method}} & \multicolumn{4}{c}{\textbf{Generation Quality} ($\uparrow$)} & \multicolumn{2}{c}{\textbf{Diversity} ($\uparrow$)} & \textbf{Novelty} ($\downarrow$) \\
    \cmidrule(lr){2-5} \cmidrule(lr){6-7} \cmidrule(lr){8-8}
     & Co-designable & PMPNN-1 Des. & PMPNN-8 Des. & pLDDT & 1 - Pairwise TM & FS Clusters & Max TM \\
    \midrule
    ProteinGenerator          & 0.11 $\pm$ 0.01 & 0.37 $\pm$ 0.02 & 0.48 $\pm$ 0.02 & 55.25 $\pm$ 0.93 & 0.55 $\pm$ 0.02 & 0.31 $\pm$ 0.07 & 0.88 $\pm$ 0.01 \\
    ProtPardelle              & 0.37 $\pm$ 0.01 & 0.36 $\pm$ 0.02 & 0.54 $\pm$ 0.01 & 63.46 $\pm$ 1.90 & 0.42 $\pm$ 0.01 & 0.08 $\pm$ 0.01 & 0.82 $\pm$ 0.00 \\
    ProtPardelle-1c           & 0.44 $\pm$ 0.02 & 0.44 $\pm$ 0.02 & 0.60 $\pm$ 0.01 & 65.92 $\pm$ 0.10 & 0.41 $\pm$ 0.01 & 0.09 $\pm$ 0.02 & 0.81 $\pm$ 0.00 \\
    CarbonNovo                & 0.62 $\pm$ 0.03 & 0.47 $\pm$ 0.01 & 0.75 $\pm$ 0.01 & 78.34 $\pm$ 2.49 & 0.59 $\pm$ 0.01 & 0.44 $\pm$ 0.02 & \textbf{0.80 $\pm$ 0.00} \\
    RFdiffusion + PMPNN & 0.43 $\pm$ 0.01 & 0.43 $\pm$ 0.01 & 0.60 $\pm$ 0.01 & 68.72 $\pm$ 0.52 & 0.61 $\pm$ 0.01 & 0.59 $\pm$ 0.03 & 0.84 $\pm$ 0.00 \\
    MultiFlow                 & 0.72 $\pm$ 0.00 & 0.64 $\pm$ 0.03 & 0.81 $\pm$ 0.02 & 79.39 $\pm$ 0.03 & 0.63 $\pm$ 0.00 & 0.56 $\pm$ 0.03 & 0.83 $\pm$ 0.00 \\
    La-Proteina (no-tri)      & 0.67 $\pm$ 0.02 & \ul{0.77 $\pm$ 0.01} & \ul{0.92 $\pm$ 0.00} & 81.89 $\pm$ 2.92 & \ul{0.63 $\pm$ 0.00} & 0.64 $\pm$ 0.02 & \ul{0.81 $\pm$ 0.01} \\
    La-Proteina (tri)         & \ul{0.77 $\pm$ 0.02} & \textbf{0.84 $\pm$ 0.04} & \textbf{0.95 $\pm$ 0.00} & \textbf{85.32 $\pm$ 2.58} & 0.59 $\pm$ 0.00 & 0.36 $\pm$ 0.01 & 0.85 $\pm$ 0.00 \\
    DPLM2                     & 0.31 $\pm$ 0.01 & 0.35 $\pm$ 0.02 & 0.49 $\pm$ 0.01 & \ul{83.69 $\pm$ 2.37} & 0.63 $\pm$ 0.02 & 0.49 $\pm$ 0.01 & 0.96 $\pm$ 0.00 \\
    ESM3 (seq$\to$str)        & 0.10 $\pm$ 0.02 & 0.14 $\pm$ 0.03 & 0.24 $\pm$ 0.02 & 60.89 $\pm$ 0.69 & 0.62 $\pm$ 0.01 & 0.70 $\pm$ 0.04 & 0.90 $\pm$ 0.01 \\
    ESM3 (str$\to$seq)        & 0.04 $\pm$ 0.01 & 0.04 $\pm$ 0.02 & 0.08 $\pm$ 0.02 & 61.10 $\pm$ 1.02 & 0.62 $\pm$ 0.02 & \textbf{0.81 $\pm$ 0.07} & 0.90 $\pm$ 0.01 \\
    \midrule
    \rowcolor{gray!20}\ours   & \textbf{0.79 $\pm$ 0.01} & 0.72 $\pm$ 0.00 & 0.91 $\pm$ 0.01 & 80.14 $\pm$ 0.31 & 0.63 $\pm$ 0.00 & 0.48 $\pm$ 0.04 & 0.83 $\pm$ 0.00 \\
    \bottomrule
\end{tabular}
}
\end{table}

\begin{table}[htbp]
\centering
\caption{\textbf{Secondary-structure composition of generated proteins.} Helix and strand percentages are reported as mean $\pm$ sample std across three independent runs.}
\label{tab:appendix_secondary_structure}
\small
\begin{tabular}{lcc}
    \toprule
    \textbf{Method} & Helix (\%) & Strand (\%) \\
    \midrule
    ProteinGenerator          & 62.09 $\pm$ 0.70 & 4.20 $\pm$ 0.32 \\
    ProtPardelle              & 52.51 $\pm$ 0.32 & 12.80 $\pm$ 0.06 \\
    ProtPardelle-1c           & 47.67 $\pm$ 0.53 & 17.83 $\pm$ 0.32 \\
    CarbonNovo                & 53.65 $\pm$ 0.53 & 22.26 $\pm$ 0.25 \\
    RFdiffusion + ProteinMPNN & 52.30 $\pm$ 0.53 & 26.70 $\pm$ 0.50 \\
    MultiFlow                 & 73.04 $\pm$ 0.87 & 10.92 $\pm$ 0.77 \\
    La-Proteina (no-tri)      & 70.82 $\pm$ 0.36 & 7.03 $\pm$ 0.15 \\
    La-Proteina (tri)         & 70.65 $\pm$ 0.71 & 6.69 $\pm$ 0.31 \\
    DPLM2                     & 40.95 $\pm$ 0.63 & 20.17 $\pm$ 0.57 \\
    ESM3 (seq$\to$str)        & 70.18 $\pm$ 0.87 & 4.66 $\pm$ 0.38 \\
    ESM3 (str$\to$seq)        & 61.56 $\pm$ 1.05 & 8.94 $\pm$ 0.41 \\
    \midrule
    \rowcolor{gray!20}\ours   & 72.19 $\pm$ 1.39 & 11.64 $\pm$ 1.00 \\
    \bottomrule
\end{tabular}
\end{table}

\begin{table}[htbp]
\centering
\caption{{Secondary-structure-stratified designability analysis.} We report the fraction of generated samples together with class-wise co-designability, with mean and std values calculated across three independent runs.}
\label{tab:appendix_secondary_structure_stratified}
\resizebox{\linewidth}{!}{
\begin{tabular}{lcccccccc}
    \toprule
    \multirow{2.5}{*}{\textbf{Method}} & \multicolumn{2}{c}{\textbf{Mainly} $\alpha$} & \multicolumn{2}{c}{\textbf{Mainly $\beta$}} & \multicolumn{2}{c}{\textbf{Mixed $\alpha/\beta$}} & \multicolumn{2}{c}{\textbf{Other}} \\
    \cmidrule(lr){2-3} \cmidrule(lr){4-5} \cmidrule(lr){6-7} \cmidrule(lr){8-9}
     & Frac. & Des. & Frac. & Des. & Frac. & Des. & Frac. & Des. \\
    \midrule
    ProteinGenerator          & 0.62 $\pm$ 0.01 & 0.16 $\pm$ 0.01 & 0.00 $\pm$ 0.00 & 0.00 & 0.12 $\pm$ 0.01 & 0.11 $\pm$ 0.04 & 0.26 $\pm$ 0.01 & 0.00 $\pm$ 0.00 \\
    ProtPardelle              & 0.38 $\pm$ 0.01 & 0.13 $\pm$ 0.01 & 0.02 $\pm$ 0.00 & 0.30 $\pm$ 0.05 & 0.50 $\pm$ 0.01 & 0.62 $\pm$ 0.01 & 0.10 $\pm$ 0.00 & 0.06 $\pm$ 0.05 \\
    ProtPardelle-1c           & 0.22 $\pm$ 0.01 & 0.60 $\pm$ 0.02 & 0.05 $\pm$ 0.00 & 0.09 $\pm$ 0.09 & 0.53 $\pm$ 0.01 & 0.55 $\pm$ 0.04 & 0.19 $\pm$ 0.01 & 0.02 $\pm$ 0.02 \\
    CarbonNovo                & 0.17 $\pm$ 0.02 & 0.82 $\pm$ 0.02 & 0.06 $\pm$ 0.01 & 0.60 $\pm$ 0.04 & 0.58 $\pm$ 0.02 & 0.62 $\pm$ 0.03 & 0.19 $\pm$ 0.01 & 0.47 $\pm$ 0.07 \\
    RFdiffusion + PMPNN & 0.17 $\pm$ 0.01 & 0.88 $\pm$ 0.05 & 0.01 $\pm$ 0.00 & 0.50 $\pm$ 0.17 & 0.76 $\pm$ 0.02 & 0.34 $\pm$ 0.01 & 0.06 $\pm$ 0.01 & 0.34 $\pm$ 0.07 \\
    MultiFlow                 & 0.44 $\pm$ 0.02 & 0.74 $\pm$ 0.02 & 0.01 $\pm$ 0.00 & 0.69 $\pm$ 0.34 & 0.30 $\pm$ 0.04 & 0.82 $\pm$ 0.01 & 0.26 $\pm$ 0.02 & 0.59 $\pm$ 0.03 \\
    La-Proteina (no-tri)      & 0.69 $\pm$ 0.01 & 0.65 $\pm$ 0.02 & 0.02 $\pm$ 0.00 & 0.19 $\pm$ 0.22 & 0.16 $\pm$ 0.01 & 0.78 $\pm$ 0.04 & 0.13 $\pm$ 0.01 & 0.72 $\pm$ 0.03 \\
    La-Proteina (tri)         & 0.69 $\pm$ 0.01 & 0.74 $\pm$ 0.02 & 0.01 $\pm$ 0.00 & 0.13 $\pm$ 0.13 & 0.20 $\pm$ 0.00 & 0.91 $\pm$ 0.03 & 0.09 $\pm$ 0.01 & 0.72 $\pm$ 0.10 \\
    DPLM2                     & 0.10 $\pm$ 0.01 & 0.36 $\pm$ 0.04 & 0.03 $\pm$ 0.00 & 0.28 $\pm$ 0.11 & 0.66 $\pm$ 0.01 & 0.30 $\pm$ 0.03 & 0.21 $\pm$ 0.01 & 0.31 $\pm$ 0.04 \\
    ESM3 (seq$\to$str)        & 0.73 $\pm$ 0.02 & 0.08 $\pm$ 0.02 & 0.02 $\pm$ 0.01 & 0.06 $\pm$ 0.10 & 0.07 $\pm$ 0.01 & 0.34 $\pm$ 0.05 & 0.17 $\pm$ 0.02 & 0.05 $\pm$ 0.05 \\
    ESM3 (str$\to$seq)        & 0.47 $\pm$ 0.03 & 0.06 $\pm$ 0.01 & 0.01 $\pm$ 0.01 & 0.00 $\pm$ 0.00 & 0.17 $\pm$ 0.01 & 0.08 $\pm$ 0.03 & 0.34 $\pm$ 0.02 & 0.01 $\pm$ 0.01 \\
    \midrule
    \rowcolor{gray!20}\ours   & 0.51 $\pm$ 0.02 & 0.76 $\pm$ 0.01 & 0.04 $\pm$ 0.01 & 0.75 $\pm$ 0.10 & 0.25 $\pm$ 0.01 & 0.85 $\pm$ 0.04 & 0.20 $\pm$ 0.00 & 0.83 $\pm$ 0.01 \\
    \bottomrule
\end{tabular}
}
\end{table}

\section{Visualization of Biomolecular Temporal Coupling}

\paragraph{Analysis of Learned Dynamics}

We visualize the learned temporal coupling for biomolecules in Fig.~\ref{fig:scheduler_comparison} and \ref{fig:more_comparison}. 
It can be seen that the structure modality is prioritized for both macromolecules (proteins) and small molecule ligands, characterized by the trajectory bowing toward the structure axis ($t_{\text{struct}} > t_{\text{seq}}$). 
This indicates that the sequence modality follows a delayed schedule, accelerating only after the coarse-grained structure takes shape.

We interpret this emergent behavior as the model identifying the most efficient path for entropy reduction. 
Geometrically, the structure determines the spatial boundary and interaction constraints as the scaffold, which contributes the most to reducing the global uncertainty of the joint distribution.
Therefore, the optimal transport plan first focuses on transporting the noisy prior to a \textit{proto-structural manifold} to establish a valid geometric context.
Once the spatial backbone is sufficiently resolved, the sequence-structure interplay becomes tractable given the informative geometric context, and then sequence identities are smoothly decoded to satisfy the chemical environments imposed by the local geometry. 
This \textit{structure-first, sequence-painting} mechanism aligns with the physical intuition of inverse folding, correcting the inductive bias of synchronous generation where the model struggles to predict sequence identities without a stabilized structural context.

\section{Sensitivity Analysis}

\paragraph{Sensitivity to Inner-loops Steps $S$ and Mixture-rate $p$}

We evaluated $S \in \{100, 400, 1600\}$, and performance is largely insensitive to $S$ as in Table~\ref{tab:sensitivity_S}, suggesting that the bottleneck is not exact inner-loop convergence, but whether the outer loop identifies a good low-cost direction. We evaluated $p \in \{0.1,0.2,0.5,0.8\}$ as in Table~\ref{tab:sensitivity_p}, and observed some variation across metrics. The role is to balance coupling guidance vs. coverage without changing the underlying objective: small $p$ weakens coupling signal (close to random), while large $p$ over-constrains training to the current candidate coupling. Intermediate $p$ achieves the best tradeoff, supporting that the mixed policy is a practical approximation for optimizing temporal coupling under finite compute.

\begin{figure*}[ht] 
    \centering
    
    \begin{minipage}[c]{0.48\textwidth}
        \centering
        \captionof{table}{Sensitivity to inner-loop steps $S$.}
        \label{tab:sensitivity_S}
        
        \resizebox{\linewidth}{!}{%
            \begin{tabular}{l|ccc|cc|cc}
                \toprule
                \multirow{2}{*}{\textbf{S}} & {Connected} & {QED} & {SA} & \multicolumn{2}{c|}{{Vina Score ($\downarrow$)}} & \multicolumn{2}{c}{{Vina Min ($\downarrow$)}} \\
                 & ($\%, \uparrow$) & ($\uparrow$) & ($\uparrow$) & Mean & Median & Mean & Median  \\ \midrule
                100  & 95.0 & 0.55 & \textbf{0.76} & -7.04 & -7.08 & -7.42 & -7.35 \\
                400  & 95.0 & 0.55 & \textbf{0.76} & \textbf{-7.31} & \textbf{-7.43} & \textbf{-7.72} & \textbf{-7.55} \\
                1600 & 95.0 & \textbf{0.56} & \textbf{0.76} & -7.08 & -7.08 & -7.44 & -7.20 \\ 
                \bottomrule
            \end{tabular}%
        }
    \end{minipage}%
    \hfill 
    \begin{minipage}[c]{0.46\textwidth}
        \centering
        \captionof{table}{Sensitivity to parameter $p$.}
        \label{tab:sensitivity_p}
        
        \resizebox{\linewidth}{!}{%
            \begin{tabular}{l|cc|cc|cc|cc}
                \toprule
                \multirow{2}{*}{\textbf{p}} & {Connected} & {PB-Valid} & {QED} & {SA} & \multicolumn{2}{c|}{{Vina Score ($\downarrow$)}} & \multicolumn{2}{c}{{Vina Min ($\downarrow$)}} \\
                 & ($\%, \uparrow$) & ($\%, \uparrow$) & ($\uparrow$) & ($\uparrow$) & Avg. & Med. & Avg. & Med.  \\ \midrule
                0.1  & 96.0 & 95.0 & 0.56 & 0.75 & -6.58 & -7.10 & -7.19 & -7.31 \\
                0.2  & 95.0 & 94.0 & 0.55 & 0.76 & \textbf{-7.31} & \textbf{-7.43} & \textbf{-7.72} & \textbf{-7.55} \\
                0.5  & 92.0 & 96.0 & 0.55 & 0.78 & -7.07 & -7.29 & -7.48 & -7.42 \\ 
                0.8  & 90.0 & 95.0 & 0.52 & 0.74 & -6.39 & -6.55 & -6.80 & -6.72 \\ 
                \bottomrule
            \end{tabular}%
        }
    \end{minipage}
\end{figure*}

\paragraph{GP Buffer Size and Kernel Choices}

We ran ablations on buffer size and GP kernels, and results are stable across settings as in Table~\ref{tab:sensitivity_buffer} and Table~\ref{tab:sensitivity_kernel}. This indicates the GP does not need to reconstruct the full cost surface precisely; it only needs to capture the topology of low-cost regions well enough to guide the path search on the loss surface, since the outer loop primarily requires a correct ordering of low-cost regions rather than exact values. We further visualize the sampled time points and observe broad, space-filling coverage at Fig.~\ref{fig:gp}

\begin{figure*}[ht] 
    \centering
    
    \begin{minipage}[c]{0.48\textwidth}
        \centering
        \captionof{table}{Sensitivity to Buffer Size.}
        \label{tab:sensitivity_buffer}
        
        \resizebox{\linewidth}{!}{%
            \begin{tabular}{l|cc|cc|cc|cc}
                \toprule
                \multirow{2}{*}{\textbf{Buffer Size}} & {Connected} & {PB-Valid} & {QED} & {SA} & \multicolumn{2}{c|}{{Vina Score ($\downarrow$)}} & \multicolumn{2}{c}{{Vina Min ($\downarrow$)}} \\
                 & ($\%, \uparrow$) & ($\%, \uparrow$) & ($\uparrow$) & ($\uparrow$) & Avg. & Med. & Avg. & Med.  \\ \midrule
                500  & 94.0 & 95.0 & 0.55 & 0.79 & -7.35 & -7.09 & -7.46 & -7.22 \\
                1000 & 95.0 & 94.0 & 0.55 & 0.76 & -7.31 & \textbf{-7.43} & \textbf{-7.72} & \textbf{-7.55} \\
                2000 & 94.0 & 95.0 & 0.53 & 0.79 & \textbf{-7.40} & -7.21 & -7.52 & -7.37 \\ 
                \bottomrule
            \end{tabular}%
        }
    \end{minipage}%
    \hfill 
    \begin{minipage}[c]{0.48\textwidth}
        \centering
        \captionof{table}{Sensitivity to GP Kernel.}
        \label{tab:sensitivity_kernel}
        
        \resizebox{\linewidth}{!}{%
            \begin{tabular}{l|cc|cc|cc|cc}
                \toprule
                \multirow{2}{*}{\textbf{GP Kernel}} & {Connected} & {PB-Valid} & {QED} & {SA} & \multicolumn{2}{c|}{{Vina Score ($\downarrow$)}} & \multicolumn{2}{c}{{Vina Min ($\downarrow$)}} \\
                 & ($\%, \uparrow$) & ($\%, \uparrow$) & ($\uparrow$) & ($\uparrow$) & Avg. & Med. & Avg. & Med.  \\ \midrule
                RBF & 95.0 & 94.0 & 0.55 & 0.76 & -7.31 & -7.43 & -7.72 & -7.55 \\
                RationalQuadratic & 94.0 & 96.0 & 0.56 & 0.76 & -6.92 & -7.22 & -7.51 & -7.34 \\ 
                \bottomrule
            \end{tabular}%
        }
    \end{minipage}
\end{figure*}

\begin{figure}[htbp]
    \centering
        \includegraphics[width=0.5\linewidth]{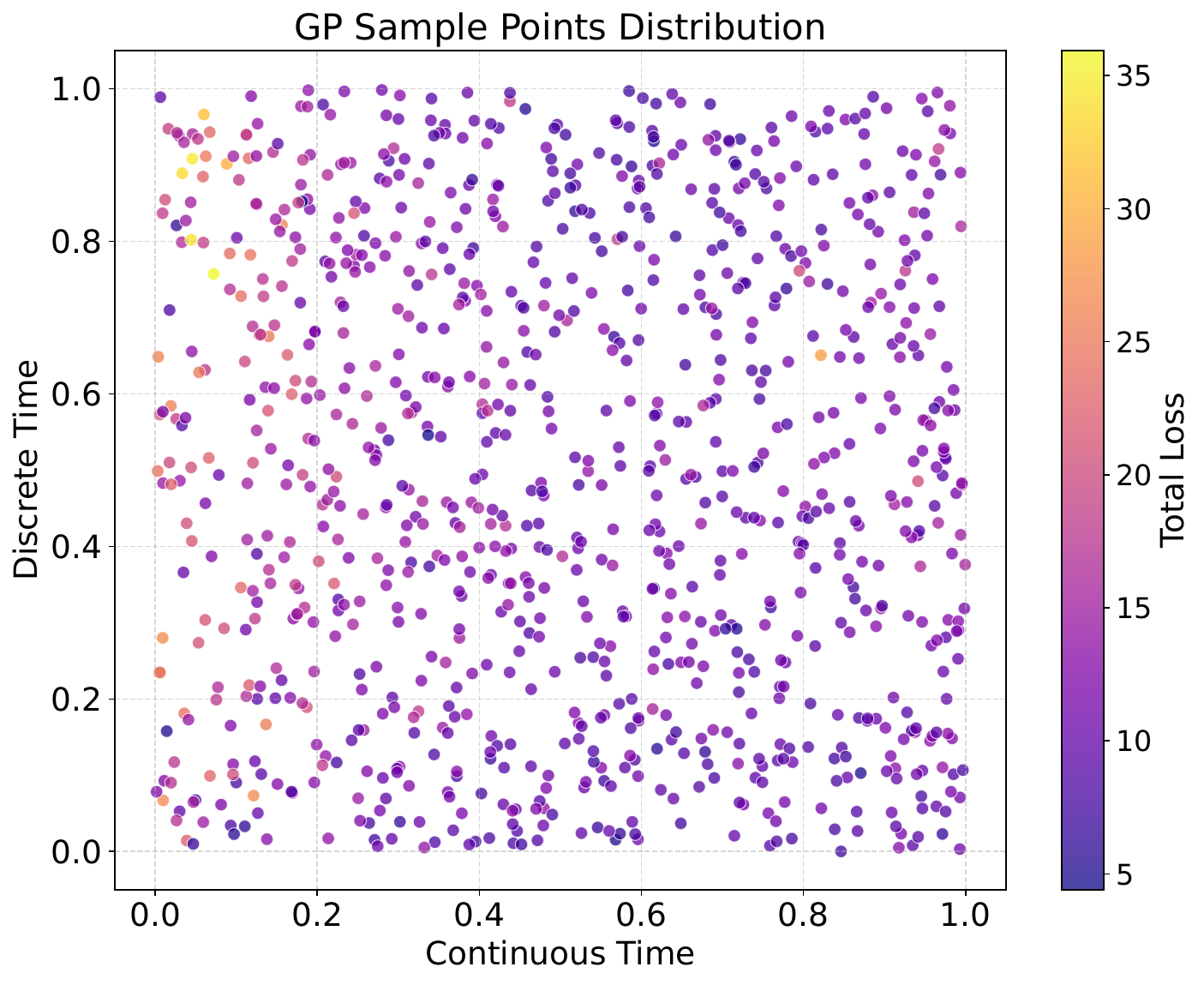}
        \caption{Planar distribution of the 1000 sampled points during the first fixed-coupling stage. The $x$-axis and $y$-axis represent Continuous Time and Discrete Time, respectively, with points colored by their corresponding Total Loss. The plot exhibits broad, space-filling coverage across the entire $[0, 1] \times [0, 1]$ search space, confirming that the Gaussian Process (GP) initialization effectively explores the parameter space without premature clustering or sampling bias.}
        \label{fig:gp}
\end{figure}

\section{Limitation}
While our results show that optimized temporal coupling consistently boosts performance for multimodal generation, our current study is focused on a clean bimodal setting, with unconditional protein co-design as the main testbed and structure-based drug design as a complementary domain, though our formulation naturally extends to more than two modalities through higher-dimensional temporal couplings. Additionally, we do not yet study conditional protein design settings such as motif scaffolding, where additional constraints may interact with the learned coupling in nontrivial ways. More broadly, although the recurring structure-leading behavior observed across both domains is encouraging, its scope remains to be characterized in regimes with different modality balances, such as intrinsically disordered proteins (IDPs) or sequence-constrained binding motifs. 











\end{document}